\documentclass[11pt]{article} %
\usepackage[margin=1in]{geometry}
\usepackage[hyphens]{url}  %
\usepackage{graphicx,natbib} %
\usepackage{graphicx}  %
\usepackage[hidelinks]{hyperref}
 \usepackage{xcolor}%

\usepackage{amsmath, amsthm, amssymb}
\usepackage{thm-restate}
\usepackage{placeins}
\usepackage{booktabs}
\usepackage[shortlabels]{enumitem}
\usepackage{subcaption}

\newcommand{\parbold}[1]{\vspace{4pt}\noindent\textbf{#1}. }
\renewcommand{\cite}{\citep}

\newcommand\pr{\text{Pr}}
\newcommand\bbR{\ensuremath{\mathbb{R}}} %
\newcommand\bbN{\ensuremath{\mathbb{N}}} %
\newcommand\bbE{\ensuremath{\mathbb{E}}} %
\newcommand\bbI{\ensuremath{\mathbb{I}}}

\newcommand\mO{\ensuremath{\mathcal{O}}}

\newcommand\mB{\ensuremath{\mathcal{B}}}

\newtheorem{remark}{Remark}
\newtheorem{definition}{Definition}

\title{Who is in Your Top Three?\\Optimizing Learning in Elections with Many Candidates}
\author{Nikhil Garg\\Stanford University\\nkgarg@stanford.edu \and Lodewijk L. Gelauff\\Stanford University\\lodewijk@stanford.edu \and Sukolsak Sakshuwong\\Stanford University\\sukolsak@stanford.edu \and Ashish Goel\\Stanford University\\ashishg@stanford.edu\\	
}
 
\begin{document}
\onecolumn
\maketitle

\begin{abstract}
\makeatletter{}%
Elections and opinion polls often have many candidates, with the aim to either rank the candidates or identify a small set of winners according to voters' preferences. In practice, voters do not provide a full ranking; instead, each voter provides their favorite $K$ candidates, potentially in ranked order. The election organizer must choose $K$ and an aggregation rule. %
	
We provide a theoretical framework to make these choices. Each $K$-Approval or $K$-partial ranking mechanism (with a corresponding positional scoring rule) induces a \textit{learning rate} for the speed at which the election recovers the asymptotic outcome. Given the voter choice distribution, the election planner can thus identify the rate optimal mechanism. Earlier work in this area provides coarse order-of-magnitude guaranties which are not sufficient to make such choices. 
Our framework further resolves questions of when randomizing between multiple mechanisms may improve learning for arbitrary voter noise models.

Finally, we use data from 5 large participatory budgeting elections that we organized across several US cities, along with other ranking data, to demonstrate the utility of our methods. In particular, we find that historically such elections have set $K$ too low and that picking the right mechanism can be the difference between identifying the ultimate winner with only a $80\%$ probability or a $99.9\%$ probability after 400 voters. %

%
	
%
%
%
%
%
%

\end{abstract}

\makeatletter{}%
\section{Introduction}
\noindent 
		Elections and opinion polls with many candidates and multiple winners are common. 
    	In participatory budgeting (PB), for example, people directly determine a part of the government's budget~\cite{alos-ferrer_two_2011,goel_knapsack_2016}. These elections often contain many candidate projects (up to 70, cf. \citet{gelauff_comparing_2018}) and only a few thousand voters, with potentially millions of dollars on the line \cite{public_agenda_public_2016}. Similarly, polls may compare tens of candidates and yet only sample hundreds of voters.
    	
	Unfortunately, the number of voters required to recover the asymptotic ranking or set of winners often scales, potentially exponentially, with the number of candidates~\cite{caragiannis_learning_2017}. Thus with many candidates, it is essential to use a voting mechanism that most efficiently elicits information from each voter. 
	
	In this work, we analyze \textit{positional scoring rules}~\cite{de1781memoire,young1975social}, mechanisms in which each position in each voter's personal ranking maps to a score given to the candidate that occupies that position. We focus on the special cases of such rules implied by $K$-Approval elicitation, in which each voter is asked to select their favorite $K$ candidates, as they the most commonly used such mechanisms in practice. Section~\ref{sec:model} formalizes our model. Then:
	
	\begin{description}
		\item  [Section~\ref{sec:learningrates}.] For a given election, we show how the particular scoring rule used affects the rate at which the final outcome (asymptotic in the number of voters) is learned. These rates, based on large deviation bounds, extend and tighten the results of~\citet{caragiannis_learning_2017}, and are precise enough to determine, for example, which of $3$-Approval and $4$-Approval is better in a particular context. We focus on the goals of learning both a ranking over all candidates and identifying a subset of winners.
		\item [Section~\ref{sec:randomization}.] Leveraging these rates, we study when randomization between scoring rules can improve learning, extending previous results to general positional scoring rules, the goal of selecting a set of winners, and arbitrary noise models. In particular, we find that randomizing between scoring rules can never speed up learning, for arbitrary noise models. This contrasts to the case when one is restricted to $K$-Approval mechanisms.
		\item [Section~\ref{sec:approvaltheory}.] For the Mallows model, we study how the optimal $K$ in $K$-Approval scales with the noise parameter, the number of candidates, and the number of winners desired. We find that, in contrast to  design choices made in practice, one should potentially ask voters to identify their favorite \textit{half} of candidates, even if the goal is to identify a \textit{single} winner.

		\item [Section~\ref{sec:data}.] We apply our approach to experimental ballots attached to real participatory budgeting elections across several US cities, as well as other ranking data from a host of domains. We find that the exact mechanism used matters: in one setting, for example, asking voters to identify their favorite candidate results in only a $80\%$ chance of identifying the best candidate after 400 voters, while asking voters for their favorite $2$ candidates identifies the same best candidate $99.9\%$ of the time. Extending our theoretical insights, we find that, historically across elections, $K$ has been set too low for effective learning. We further identify real-world examples in which randomization would have sped up learning. 
		
	\end{description}   
	
	Our work bridges a gap between coarse theoretical analyses of voting rules and the fine-grained design questions a practitioner wishes to answer. Proofs are in the Appendix. 
\makeatletter{}%
\section{Related work}
\label{sec:relatedwork}
Our work is part of several strands of research on mechanisms that elicit peoples' preferences. Aggregating voter rankings has a long history~\cite{de1781memoire,marquis1785essai,copeland1951reasonable,kemeny1959mathematics,young1988condorcet}.

\parbold{Learning properties of voting rules} Most related are works that study the learning properties of voting rules, assuming that a ``true'' ranking. One approach is to specify a noise model under which voter preferences are drawn (e.g., Mallows, Plackett-Luce) and then derive error rates by the number of voters for maximum likelihood or similar estimators under the model~\cite{maystre_fast_2015,zhao_learning_2016,lu_learning_2011,guiver_bayesian_2009,procaccia_is_2015,de2016minimising,chierichetti2014voting}. 

\citet{caragiannis_when_2013} ask similar questions to us: under what voter noise models do certain voting rules asymptotically recover the true underlying ranking, and how quickly do they do so. They define a class of voting rules and voter noise models under which a ``true'' ranking of candidates is eventually recovered. They further show that for a subset of this class (that does not contain positional scoring rules) and under the Mallows model, only a number of voters that is logarithmic in the number of candidates is required, where each voter provides a full ranking. \citet{lee2014crowdsourcing} develop an algorithm that can approximate the Borda rule, given a number of comparisons by each voter that is logarithmic in the number of candidates.

Most similar is that of~\citet{caragiannis_learning_2017}. They show that under the Mallows model, $K$-Approval with any fixed $K$ takes exponentially many voters (in the number of candidates) to recover the underlying ranking; on the other hand, $K$-approval with $K$ chosen uniformly at random for each voter takes only a polynomial number of voters.

These works provide order estimates for the learning rate, \textit{asymptotic in the number of candidates}; fine-grained differentiation between different rules or $K$-Approval mechanisms for a given election is not possible. We provide the latter and show that it matters.

  \parbold{Other approaches to comparing mechanisms} Many works take an axiomatic and computational approach, comparing mechanisms that may produce different outcomes even given asymptotically many votes~\cite{fishburn_bordas_1976,fishburn_axioms_1978,staring_two_1986,tataru_relationship_1997,wiseman_approval_2000,ratliff_startling_2003,elkind_properties_2015,aziz_computational_2015,caragiannis_subset_2017,aziz_justified_2017,lackner_consistent_2018,lackner_quantitative_2018,faliszewski_framework_2018}.~\citet{caragiannis2019optimizing} for example show how to find a scoring rule that most agrees with a given partial ground truth ranking. In contrast, we compare mechanisms' learning rates under a condition (formalized in Section~\ref{sec:consistency}) in which they produce the same asymptotic outcome.

  \citet{benade_efciency_2018}~and~\citet{gelauff_comparing_2018} experimentally compare different mechanisms across several dimensions, including ease of use and consistency with another mechanism; the latter leverages data from a participatory budgeting election at a university.

 \parbold{Large deviation analysis of elicitation mechanisms} Theoretically, we leverage large deviation rates and Chernoff bounds to derive how quickly a given scoring rule learns its outcome; see work of \citet{dembo_large_2010} for an introduction to large deviations. This work is thus conceptually similar to work on elicitation design for rating systems
 ~\cite{garg2018designing,garg2019designing}. In those works, the authors derive large deviation-based learning rates that depend on the questions that are asked to buyers as they review an item, where the goal is to accurately rank items; they further run an experiment on an online labor platform. In that setting, however, buyers rate a \textit{single} item, and mechanisms are distinct based on the \textit{behavior} they induce; in this work, voters see all the candidates and provide a partial ordering, and different designs (e.g., $3$-Approval vs $4$-Approval) \textit{constrain} the types of orderings voters can provide. %

\makeatletter{}%
\section{Model}
\label{sec:model}
We now present our model and a condition under which different positional scoring rules induce the same asymptotic outcome.

\subsection{Model primitives}

We begin with the model primitives: candidates and voters, the election goal, and elicitation and aggregation.   

\parbold{Candidates and Voters} There is a set of $M$ candidates $C = \{1,\dots,M\}$, typically indexed by $i,j\in C$. There are $N$ voters $V = \{1,\dots,N\}$. Each voter $v\in V$ has a strict ranking of candidates $\sigma_v$, drawn independently and identically from probability mass function over strict rankings $F(\sigma)$. Let $i \succ_\sigma j$ denote that $i$ is preferred over $j$ in $\sigma$, and $\sigma(i) = k$ denote that candidate $i$ is in the $k$th position in $\sigma$.

A special case for $F$ is the \textit{Mallows} model \cite{mallows1957non}, in which there is a ``true'' societal preference $\sigma^*$ from which each voter's ranking is a noisy sample. In particular, 
\[
F_\text{Mallows}(\sigma) \propto {\phi^{d(\sigma, \sigma^*)}}
\]
Where $d(\sigma, \sigma^*)$ is the Kendall's $\tau$ distance between rankings $\sigma, \sigma^*$, and $\phi \in [0,1]$ is the noise parameter: the smaller it is, the more concentrated $F$ is around $\sigma^*$.

\parbold{Election goal} We assume that the \textit{goal} $G$ is to divide the candidates into $T$  disjoint, ordered \textit{tiers} $G = \{C_1, \dots , C_T\}$, such that $C = \cup_{t=1}^T C_t$, where candidate $i \in C_{s}$ is deemed societally preferable over $j\in C_{t}$ if $s<t$. The size of each tier is fixed before the election. For example, recovering a strict ranking over all candidates corresponds to $G = \{C_1, \dots, C_M\}$, where $|C_t| = 1$. Alternatively, identifying a set of $W$ winners, without distinguishing amongst the winners, corresponds to $G = \{C_1, C_2\}$, with $|C_1| = W$.

In the main text and especially the empirics, we will focus on the task of selecting $W$ winners as it is the most common task in practice. However, this general notation allows comparison of the learning properties of different settings, and for example ask how much more expensive is it (in terms of the number of voters needed) to identify a strict ranking as opposed to just a set of winners.

\parbold{Elicitation and Aggregation}
Voters vote using an elicitation mechanism. Their votes are then aggregated using a positional scoring rule, parameterized as $\beta: \{1,\dots,M\} \mapsto \bbR$. We consider the following mechanisms:
\begin{description}
	\item [$K$-Ranking] Voter $v$ ranks her favorite $K$ candidates, i.e., reveals $\{(i, {\sigma_v}(i)) \,:\, {\sigma_v}(i) \leq K \}$. Candidate $i$ then receives a score $s_{iv} = \beta({\sigma_v}(i))$ if ranked, $0$ otherwise. For example, $\beta(k) = M-k$ for the Borda count.\footnote{In Borda, candidates not ranked receive a score $(M-K-1)/2$, consistent with assuming they are all tied in position $(K + 1)$.}
	
	\item [$K$-Approval] Voter $v$ selects her favorite $K$ candidates, i.e., reveals $\{i \,:\, {\sigma_v}(i) \leq K \}$. A candidate receives a score $s_{iv} = 1$ for being selected, $0$ otherwise. 

\end{description}

$\beta$ encodes both elicitation and aggregation. For example, $K$-Approval is equivalent to $K$-ranking with score function $\beta(k) = \bbI[k \leq K]$.  Furthermore, note that given $K$-ranking data, one can simulate $K'$-ranking elicitation for $K'\leq K$ with a $\beta$ s.t. $\beta(k) = 0$ for $k>K'$. %

The scoring rule $\beta$ is a design choice made by the election organizer, and so we will refer to $\beta$ as the election's \textit{design}. We restrict ourselves to non-constant, non-increasing scoring rules, i.e.,  $\beta \in \mB = \{\beta : \forall k<\ell \in {1,\dots,M}, \beta(k)\geq \beta(\ell), \text{ and } \exists k<\ell, \beta(k)> \beta(\ell) \}$.

\parbold{Outcome} After $N$ voters, candidate $i$'s cumulative score is $s_i^N = \frac{1}{N} \sum_{v=1}^N s_{iv}$. Candidates are ranked in descending order of score, to form ranking ${\sigma}^N$, with ties broken uniformly at random. We denote the \textit{outcome} after $N$ voters, corresponding to the goal $G$, as $O^N(M, F, \beta, G)$. For example, for the goal of selecting $W$ winners, $O^N(M, F, \beta, G)$ is simply the top $W$ candidates in ${\sigma}^N$. When $(M, F, \beta, G)$ is clear from context, we will refer to the outcome as $O^N$.

 As the number of voters $N\to\infty$, candidate scores $s_i^N \to \bbE_F[s_{iv}] \triangleq s_i$ by the law of large numbers; when such expected scores are distinct, i.e., $s_i \neq s_j$ for $i\neq j$, then $\sigma^N \to \sigma^*$ for some ranking $\sigma^*$. However, note that there may exist an asymptotic outcome $O^N \to O^*$ even without an asymptotic ranking $\sigma^N \to \sigma^*$, as long as expected scores $s_i$ and goal $G$ are such that candidates with identical expected scores are sorted into the same tier.  

\makeatletter{}%
\subsection{Asymptotic design invariance}
\label{sec:consistency}
The asymptotic outcome $O^*$ of an election may vary with the scoring rule $\beta$. For example, there may be a different winner if voters are asked to identify their favorite two candidates than if they identify their single favorite candidate, if the winner in the latter case is a polarizing candidate. As an axiomatic comparison between outcomes is out of the scope of this paper, we restrict our attention to cases where all ``reasonable'' choices of different $\beta$ asymptotically result in the same outcome (where ``reasonable'' corresponds to the set of scoring rules $\mB$ defined above).

\begin{definition}
	A setting $(M, F)$ is \textit{asymptotically design-invariant} for goal $G$ if any reasonable $\beta$ induces the same outcome asymptotically. $\exists O^*: \forall \beta \in \mB,$
	\[ \lim_{N \to \infty} O^N(M, F, \beta, G) = O^*, \text{ with probability } 1 \]
\end{definition}

Such design invariance only occurs under a fairly strong condition on the voter preference distribution: that the candidates can be separated into tiers (according to goal $G$) such that candidates in higher tiers are \textit{strictly} more likely to be ranked by a voter in the top $k$ positions, for all $k<M$, than are candidates in lower tiers. 

\begin{restatable}{proposition}{thmconsistency}
	\label{thm:consistency} 
	A setting $(M, F)$ for goal $G$ is asymptotically design-invariant if and only if there exist candidate tiers $O^* = \{C_1^*\dots C_T^*\}$ (corresponding to $G$) s.t. $\forall s<t$: $i \in C^*_s, j\in C^*_t \implies$ $\pr_F(\sigma_v(i) \leq k) > \pr_F(\sigma_v(j) \leq k)$, $\forall k\in\{1 \dots M-1 \}$. 
\end{restatable}

Note that this condition is stronger than stochastic dominance as the inequality is strict for \textit{every} position $k$. 

 This proposition connects to~\citet{caragiannis_when_2013} as follows: they prove that many rules (including all positional scoring rules and the Bucklin rule) asymptotically recover the base ranking $\sigma^*$ of a generalization of the Mallows model in which the probability $F(\sigma)$ of a ranking $\sigma$ is monotonic in the distance $d(\sigma, \sigma^*)$, where distance function $d$ is itself in some general class that contains the Kendall's $\tau$ distance. Their results directly imply that such noise models, including the standard Mallows model, are asymptotically design-invariant for any goal $G$. 

However, for goals $G$ where recovering a full ranking is unnecessary, the condition in Proposition~\ref{thm:consistency} is weaker than the assumptions of~\citet{caragiannis_when_2013}; there need not even be a single base ranking $\sigma^*$. For example, when $G$ such that we wish to select a set of $W$ winners, $F$ corresponding to a mixture of Mallows models -- with all possible permutations of the $W$ candidates in the top $W$ positions in the base rankings -- would still be design-invariant. Constructing a general class of ranking noise models that satisfies this property is an avenue for future work. 

Assuming asymptotic design-invariance on voter preferences $F$ may seem restrictive. However, absent axioms -- that are precise enough for design purposes --  to prefer one scoring rule $\beta$ over another, the assumption allows us to proceed in a principled manner. We believe it is unlikely that such precise, satisfactory axioms exist generally. In the Appendix, we provide a simple example (similar to that of \citet{staring_two_1986}) where $1$-Approval and $2$-Approval select disjoint sets of $2$-Winners, and such examples can be adapted more generally to selecting $W$ winners from either $K$-Approvals or $K'$-Approvals. In participatory budgeting with the goal of identifying 6-10 winning projects out of over twenty projects, it is unclear whether 
there is a principled reason to prefer $4$-Approval over $8$-Approval. However, such axioms would be an interesting avenue for future work. 

Furthermore, in Section~\ref{sec:modelvalidation} we show that design invariance is often approximately satisfied in practice, especially for identifying a small set of winners, using data from a wide range of participatory budgeting and other elections.

\makeatletter{}%
\section{Learning Rates and Optimal Design}
\label{sec:learningrates}

Different elicitation and aggregation mechanisms may take different amounts of voters to learn the asymptotic outcome. For example, suppose we want to identify the worst candidate out of 100, where the voter's rankings are drawn from a Mallows model with $\phi > 0$. Then, asking each voter to identify their single favorite candidate will eventually identify the worst candidate, but after many more voters than if we ask each voter to identify their least favorite candidate. We make such learning rates precise in this section. Our results in this section extend those of~\citet{caragiannis_learning_2017} as discussed above, both to arbitrary positional scoring rules and by providing tighter bounds for how a scoring rule affects the convergence rate. These rates are precise enough to \textit{design} scoring rules, for example comparing $4$-Approval and $8$-Approval in the above example.

\subsection{Learning rates}

We begin by deriving rates for how quickly a given positional scoring rule $\beta$ learns its asymptotic outcome $O^*$ (given it exists), as a function of the voter preference model $F$. %
In particular, we use \textit{large deviation rates} at which a scoring rule learns~\cite{dembo_large_2010}. 
\begin{definition}
	Consider a sequence $\{A_N \geq 0\}_{N \in \bbN}$, where $A_N \to 0$ . Value $r>0$ is the \textit{large deviation rate} for $A_N$ if
	\begin{align*}
	r = -\lim_{N\to\infty}\frac{1}{N} \log A_N
	\end{align*}
\end{definition} When $r>0$ exists, $A_N \to 0$ exponentially fast, with exponent $r$ asymptotically, i.e., $A_N$ is $e^{-rN \pm o(N)}$. These rates provide us both upper and lower bounds for the probability of an error or the number of errors in an outcome after $N$ voters, up to polynomial factors. In particular, in the propositions below, we will calculate the large deviation rate of errors in the outcome. We will also then provide (loose) upper bounds for such errors after $N$ voters that hold without any missing polynomial factors, for any $N$. These upper bounds are equivalent to Chernoff bounds. 

The particular forms for these rates, derived below for general noise models $F$, may seem complex. However, they are useful  both for theoreticians and practitioners. For example, in Section~\ref{sec:randomization}, we use the structure of such rates to resolve open questions regarding when randomization between mechanisms can help learn the outcome from votes drawn from an \textit{arbitrary} noise model. In Section~\ref{sec:data}, we show that learning rates -- even when empirically calculated -- reflect the true behavior of errors in real elections with a small number of voters; we then use empirically calculated learning rates to draw design insights across elections.

\parbold{Rates for separating two candidates} We now derive the large deviation learning rates for recovering the true ordering between a pair of candidates $i,j$, given noise model $F$. These rates will directly translate to the learning rate for the overall election, given some goal $G$.

\begin{restatable}{proposition}{thmpairwiselearning}
	\label{thm:pairwiselearning}
	Fix scoring rule $\beta\in\mB$, voter distribution $F$, and consider candidates $i,j$ such that $s_i > s_j$. Then, the probability of making a mistake in ranking these two candidates after $N$ voters, $\pr(\sigma^N(i)>\sigma^N(j))$, goes to zero with large deviation rate
	\[r_{ij}(\beta) = -\inf_{z\in\bbR} \log \bbE_F\left[\exp\left(z\left(\beta(\sigma_v(i)) - \beta(\sigma_v(j))\right)\right)\right]\] 
	Further, the following upper bound holds for any $N$. 
	\[\pr(\sigma^N(i)>\sigma^N(j)) \leq \exp(-r_{ij}(\beta)N)\]
\end{restatable}
The proof follows directly from writing a random variable for the event of making a mistake after $N$ voters and then applying known large deviation rates. This simplicity emerges because positional scoring rules are additive across voters.

The proposition establishes that -- for a fixed number of candidates $M$ and voter noise model $F$ -- the probability of making a mistake on any single pair of candidates $i,j$ decreases exponentially with the number of voters, at a rate governed by the scoring rule $\beta$ and the candidates' relative probabilities of appearing at each position of a voter's preference ranking. The rate $r_{ij}(\beta)$  is non-negative, and and larger values correspond to faster learning of the relative ranking of $i,j$. Note that for notational convenience, we suppress $F$ in the argument for the rate. 

 For general $\beta$, we cannot find a closed form for $r_{ij}(\beta)$. However, the structure of this rate, in particular that of the argument in the $\log(\cdot)$, will directly let us show that randomization cannot help learning outcomes among positional scoring rules, for \textit{arbitrary} noise models $F$.

  For $K$-Approval voting, further, the rate simplifies.

\begin{restatable}{proposition}{lempairwiseapproval}
	\label{lem:pairwiselearning_approval}
	Consider $\beta$ consistent with $K$-Approval voting for some fixed $K$, and candidates $i,j$ such that $s_i > s_j$. Then the large deviation rate $r_{ij}(\beta)$ in Proposition~\ref{thm:pairwiselearning} is
	\begin{align*}
	r_{ij}(K) &= -\log\left( 2\sqrt{t_{ij}^i(K)t_{ij}^j(K)} + 1 - t_{ij}^i(K) - t_{ij}^j(K) \right)
	\end{align*}
	Where $t^i_{ij}(K) \triangleq \pr_F(\sigma_v(i) \leq K, \sigma_v(j) > K)$, i.e., the probability that a voter approves $i$ but not $j$.
\end{restatable}

The proof follows directly from the structure of $\beta$ for $K$-Approval, $\beta(k) = \bbI[k\leq K]$; for each pair of candidates, the sufficient statistics are how often each candidate appears in a voter's top $K$ list but the other candidate does not. 

We overload notation and use $K$ directly in the argument for $r_{ij}(K)$. This rate function $r_{ij}(K)$ is convex in the probabilities $t^i_{ij}(K), t^j_{ij}(K)$; this fact will let us show that randomization, even among $K$-Approval mechanisms, cannot help learning the relationship of any pair of candidates. 

\parbold{Rates for learning the outcome}
In general, the rates at which one learns each pair of candidates immediately translate to rates for learning the entire outcome $O^*$. 

\begin{restatable}{proposition}{thmgoallearning}
	\label{thm:goal_learning}
	Consider goal $G$ and $\beta\in\mB$ such that $O^N \to O^*$. Let $Q^N$ be the expected number of errors in the outcome after $N$ voters, $\sum_{i\in C_s^*,j\in C_t^*, s < t} \pr(\sigma^N(i)>\sigma^N(j))$. Then $Q^N$ goes to zero with large deviation rate
	\[
	r(\beta) = \min_{i\in C_s^*,j\in C_t^*, s < t} r_{ij}(\beta) \label{eqnpart:outcomerate}
	\]
	Further, the following upper bound holds for any $N$. \[Q^N \leq M^2 \exp(-rN)\] %
\end{restatable}

The large deviation rate $r(\beta)$ thus provides a tight characterization for how many voters it takes to (with high confidence) recover the asymptotic outcome of an election. Note that the goal plays an important role: for selecting $W$ winners, for example, it is not important to learn the exact relationship among candidates $\{1, \dots, W \}$, speeding up outcome learning. Design $\beta$ also matters; e.g., even amongst approval voting mechanisms, $K=1$ vs $K=5$ will produce substantially different $t_{ij}^i(K)$. To derive learning rates for $K$-Approval for any given noise model or using real-world data, one simply needs to calculate these values. We do so numerically for the Mallows model and empirically with real world data in Sections~\ref{sec:approvaltheory} and~\ref{sec:data}, respectively.

\subsection{Optimal design and discussion}
Now that we can quantify how quickly a given scoring rule $\beta$ learns its asymptotic outcome, we apply our framework to \textit{designing} elections, i.e., choosing an optimal scoring rule $\beta$. For the rest of this work, we assume that the setting $(M,F)$ is asymptotically design-invariant for the goal $G$, i.e., there exists an outcome that is asymptotically induced by every reasonable scoring rule. Then, the design of an election $\beta$ only affects the \textit{rate} at which the election converges to the asymptotic outcome $O^*$, as calculated above. With no other constraints, then, the design challenge is simple: find the rate optimal $\beta$.
\begin{definition}
	A scoring rule $\beta^* \in \mB$ is \textit{rate optimal} if it maximizes the rate in Proposition~\ref{thm:goal_learning}. %
	$K^*$-Approval is \textit{Approval rate optimal} if it maximizes the rate among $K$-Approval mechanisms.
\end{definition}

Rate optimal designs $\beta$ learn the outcome faster than others in the number of voters, and so are preferable to other designs. What influences how quickly a design $\beta$ learns? 

$\bbE_F\left[\exp\left(z\left(\beta(\sigma_v(i)) - \beta(\sigma_v(j))\right)\right)\right]$ must be small (near zero) for negative $z$, and so $\beta(k) - \beta(k')$ must be large when $\pr(\sigma_v(i) = k, \sigma_v(j) = k')$ is large. In other words, a scoring rule must reward a candidate achieving a position in a voter's ranking that is only achieved by asymptotically high-ranking candidates. For example, if it is common for worse candidates to be ranked second in a given voter's ranking but not to be ranked first, then $\beta(1)\gg \beta(2)$ would be beneficial.

Note that finding such designs requires knowledge of the voter noise model $F$, which in many settings may not be available before the election. However, next in Sections~\ref{sec:theoreticalinsights}~and~\ref{sec:data}, we show that there are valuable insights that apply \textit{across} elections, including how our approach has informed participatory budgeting deployments. %

\makeatletter{}%

\section{Theoretical Design Insights}
\label{sec:theoreticalinsights}
\begin{figure*}
	\centering
	\begin{subfigure}[b]{.56\textwidth}
		\includegraphics[width=\linewidth]{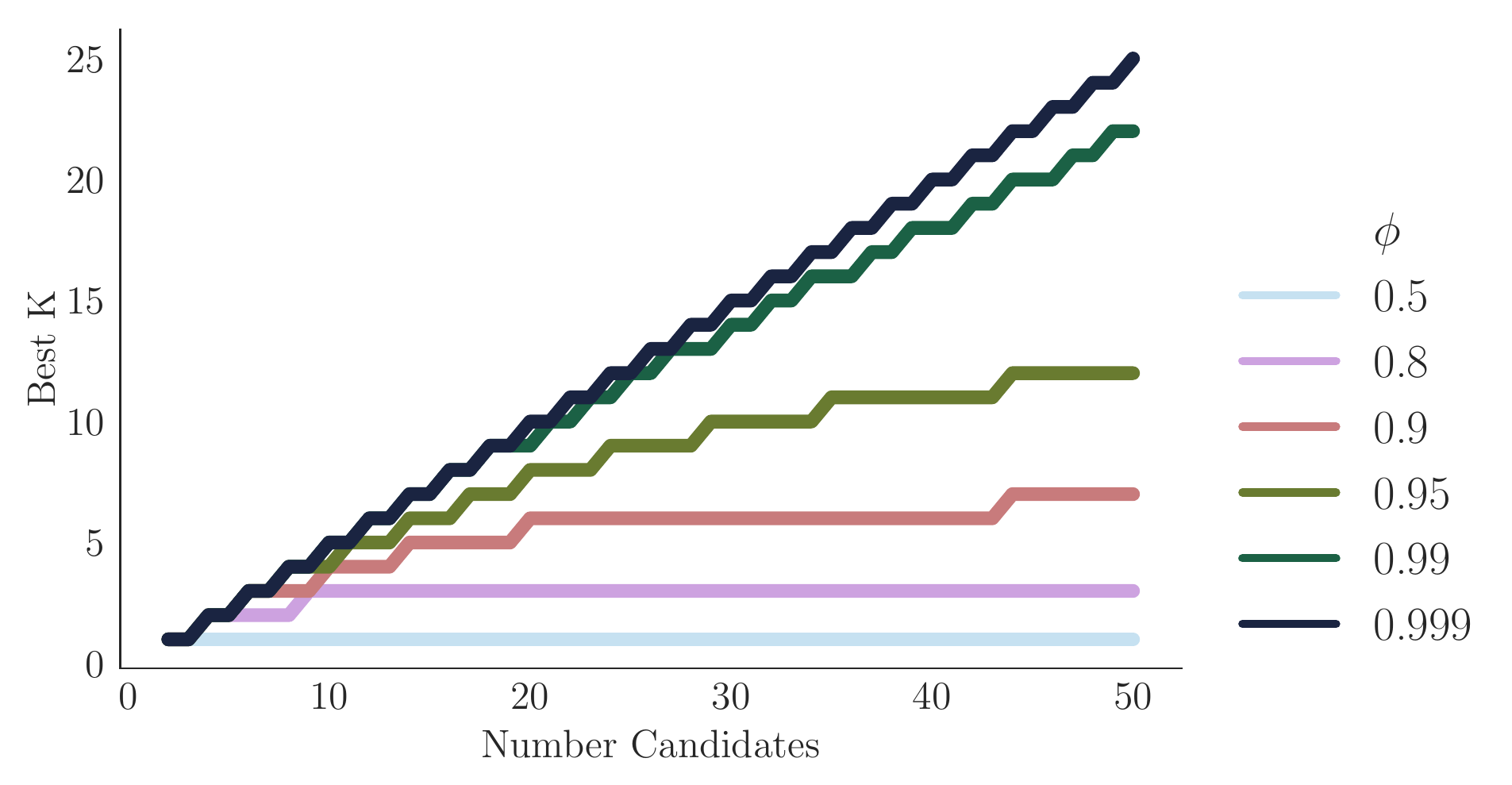}\vfill
		\caption{For selecting $W=1$ winner as number of candidates vary.}
		\label{fig:mallowsnumcandid}
	\end{subfigure}\hfill
	\begin{subfigure}[b]{.44\textwidth}
		\includegraphics[width=\linewidth]{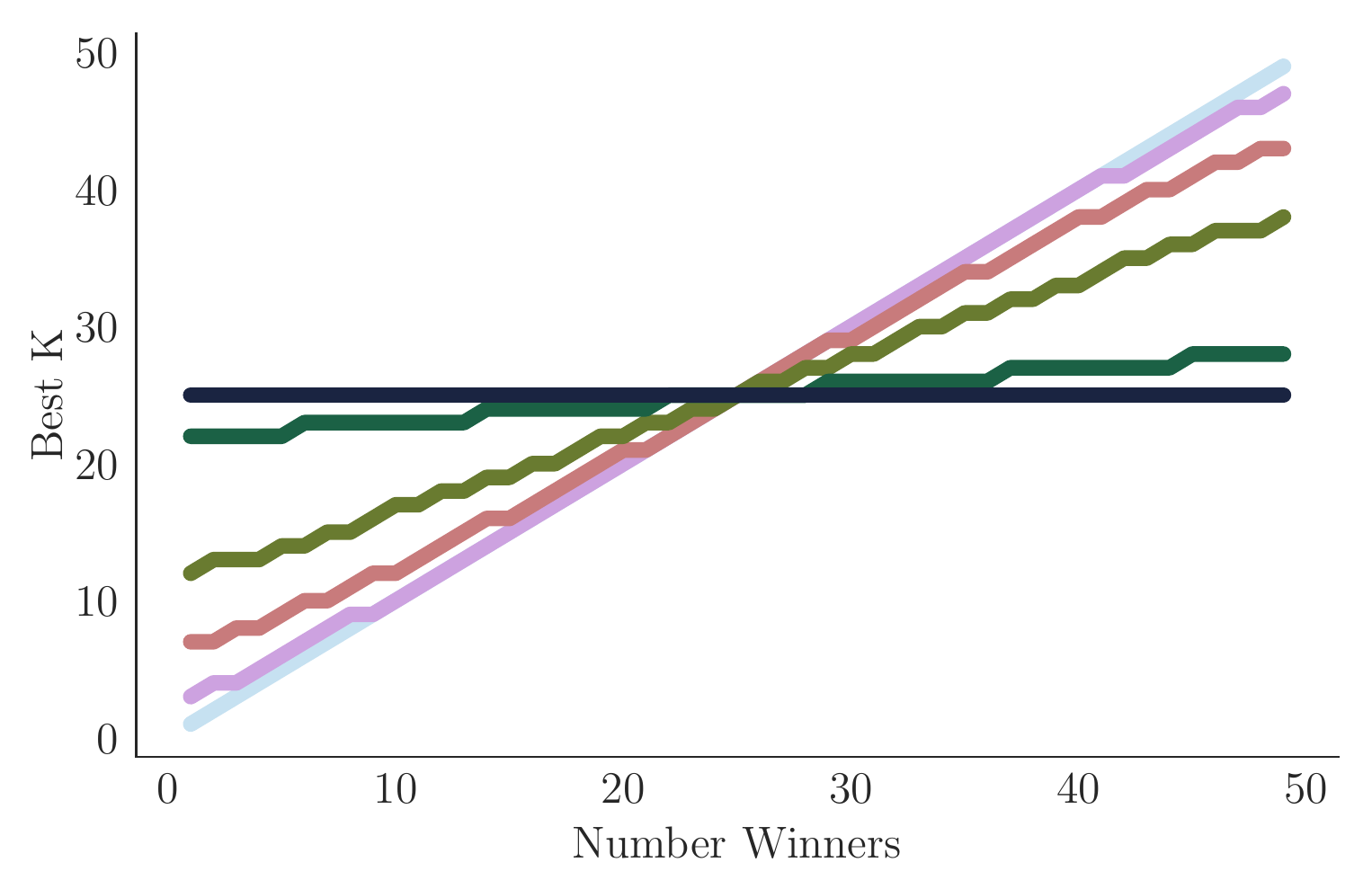}\vfill
		\caption{For $M=50$ candidates as number of winners vary.}
		\label{fig:mallowsnumwinners}
	\end{subfigure}\hfill
	\caption{$K$-Approval rate optimal mechanism for the Mallows model as $\phi$, number of candidates, and number of winners vary. }
	\label{fig:mallowsplots}
\end{figure*}
The learning rates derived in the previous section provide election design insights, even before our approach is applied to real-world data. In particular, in this section, we first extend the previous literature on the (potential) benefits of randomizing between mechanisms. Then, we study the task of selecting $W$ winners using $K$-Approval voting.

\subsection{When does randomization help?} 
\label{sec:randomization}

We now consider the question of whether \textit{randomizing} between mechanisms in an election may speed up learning. By randomization, we mean: consider a set of scoring rules $B = \{\beta_1, \dots, \beta_P\} \subseteq \mB$; elicitation and aggregation for a given voter is done according to a scoring rule picked at random from $B$, where $\beta_p$ is selected with probability $d_p$.

Note that the learning rate of such randomized schemes can be calculated as before, by summing across $\beta_p$ inside the $\bbE[\cdot]$ of $r_{ij}(\beta)$ or -- for $B$ consisting only of $K$-Approval votes -- directly through the resulting probability that the voter approves $i$ but not $j$. We use $r_{ij}(B,D), r(B,D)$ to denote the candidate pairwise and overall outcome learning rates, respectively, for randomized mechanism $(B,D)$, where $B = \{\beta_1, \dots, \beta_P\} \subseteq \mB$ and $D = \{d_1, \dots, d_P \}$. 

It is known that in some settings randomization improves learning, asymptotically in the number of candidates.  \citet{caragiannis_learning_2017} provide an example in which randomizing uniformly between all possible $K$-Approval mechanisms outperforms any static $K$-Approval elicitation, when the goal is to rank all the candidates. Their insight is that, under the Mallows model and under a fixed $K$, either the first two candidates will be hard to distinguish from each other, or the last two will, and randomizing between mechanisms balances learning each pair. 

We now study randomization for the goal of selecting $W$ winners and for arbitrary positional scoring rules and voter noise models.
Our first result is that randomizing between scoring rules does not help, for \textit{any} voter noise model, in contrast to the case when restricted to approval votes.

\begin{restatable}{theorem}{lemrandomizebetterscoring}
	\label{lem:randomizebetterscoring}
	Randomization does not improve the outcome learning rate for any asymptotically design-invariant noise model $F$ or goal $G$.  For any randomized scoring rule mechanism $(B,D)$, where $B\subset \mB$, for any $F$, $G$, the scoring rule $\beta^*(k) = \sum_{p} d_p \beta_p(k)$ satisfies $r(\beta^*) \geq r(B,D)$.

\end{restatable}

The result follows from the fact that $\bbE_F\left[\exp\left(z\left(\beta(\sigma_v(i)) - \beta(\sigma_v(j))\right)\right)\right]$  is convex in $\beta(k)$, for all $i,j,z,F$. Then, given a randomization over $\beta_1, \dots \beta_P$, we can increase $-\inf_z \log(\cdot)$ by decreasing its argument, by instead using the static scoring rule defined by the corresponding convex combination of $\beta_1, \dots \beta_P$. Note that such a negative result cannot be obtained via analysis that is asymptotic in the number of candidates; we need learning rates for a given election.

Next, we further refine the result of~\citet{caragiannis_learning_2017}, by showing that the ``pivotal pair'' feature of their example -- where different pairs of candidates dominate the learning rate for different mechanisms -- is key. In particular, our next result establishes, again for any noise model, that randomization amongst $K$-Approval mechanisms cannot help separate any \textit{given} pair of candidates.

\begin{restatable}{theorem}{lemrandomizenotbetterapproval}
	\label{lem:randomizenotbetterapproval}
	
		Randomization amongst $K$-Approval mechanisms does not improve the learning rate for separating a given pair of candidates $i,j$ for any asymptotically design-invariant noise model $F$ or goal $G$. For any randomized $K$-Approval mechanism $(B,D)$, where $\beta_p\in B$ corresponds to $p$-Approval, for any $F$, $G$, there exists a mechanism $K_{ij}^*$-Approval such that $r_{ij}(K_{ij}^*) \geq r_{ij}(B,D)$.

\end{restatable}

The proof relies on the pairwise rate function $r_{ij}(K)$ being convex in the approval probabilities $t^i_{ij}(K), t^j_{ij}(K)$. 

This theorem directly implies that, for the Mallows model, randomization among $K$-Approval voting cannot speed up learning when the goal is to identify a set of $W$ winners, as opposed to when the goal is to rank. 

\begin{restatable}{corollary}{lemmallowsnorando}
	\label{lem:mallowsnorando}
			Randomization among $K$-Approval mechanisms does not improve the learning rate for selecting $W$ winners from the Mallows model. For any randomized $K$-Approval mechanism $(B,D)$, where $\beta_p\in B$ corresponds to $p$-Approval, for selecting $W$ winners from the Mallows model, there exists an Approval rate optimal mechanism $K^*$-Approval such that $r(K^*) \geq r(B,D)$.
	
\end{restatable}

The proof simply notes that under the Mallows model with this goal, the candidate pair $W, W+1$ (when candidates are indexed according to reference distribution $\sigma^*$) is pivotal regardless of the $K$-Approval mechanism used. 
 This corollary does not extend to arbitrary noise models, where randomization amongst $K$-approval mechanisms may improve the learning rate.

\begin{restatable}{theorem}{lemrandomizebetterapprovalWselection}
	\label{lem:randomizebetterapproval_Wselection}
			Randomization among $K$-Approval mechanisms may improve the learning rate for the goal of selecting $W$ winners. There exist asymptotically design-invariant settings $(M,F)$ for the goal of selecting $W$ winners such that a randomized $K$-Approval mechanism $(B,D)$, where $\beta_p\in B$ corresponds to $p$-Approval, satisfies \[r(B,D) > \max_K r(K) \] 
\end{restatable}

We prove the result two ways: (1) we construct an example in which candidate $h$ is asymptotically selected, and candidates $i,j$ are not. Which of $h \succ i$ or $h \succ j$ is the pivotal pair (determines the overall rate function) depends on the $K$-Approval mechanism used, and randomizing between two mechanisms improves the overall rate; (2) perhaps more interestingly, we find many examples in our real PB elections and other ranking data in which randomization would have sped up learning for the task of selecting a set of winning candidates (see Section~\ref{sec:randomizationpractice}).

\subsection{$K$-Approval for selecting $W$ winners}
\label{sec:approvaltheory}
One of the most common voting settings is identifying a set of $W$ winners using $K$-Approval, whether in representative democracy elections (typically $K=W=1$), polling for such elections (where the goal often is to identify the top few candidates out of many, especially in primary races), or crowd-sourcing labels (where one wants one or a few labels for an item out of many possible ones). Here, we study how to design such elections, i.e., how to choose the best $K$, i.e., the one that maximizes the learning rate. For simplicity, we work with the Mallows model, extending the resulting insights to real-world data in the next section.

Recall that in a Mallows model, each voter's ranking is a noisy sample from a reference distribution $\sigma^*$. With this symmetric model, one may believe that setting $K=W$ is always optimal. For example, when noise parameter $\phi = 0$ and so each voter's ranking is exactly $\sigma^*$, $K=W$ is optimal; in fact, any other design $K \neq W$ fails to correctly identify the set of winners even asymptotically: it would not distinguish among the first $K$ candidates in $\sigma^*$ or among the last $M-K$ candidates. However, our next result establishes that the cases with $\phi > 0$ are different. 

\begin{restatable}{theorem}{lemmallowsnotWK}
	\label{lem:mallowsnotWK}
	Under the Mallows model and the goal of selecting $W$ winners, $W$-Approval may not be Approval rate optimal.
\end{restatable}

We prove the theorem by example. To find this example and to generate the plots discussed next, we use an efficient dynamic program to exactly calculate the joint distributions of the locations $\sigma_v(i),\sigma_v(j)$ of pairs of candidates $i,j$ in a voter's ranking, given the Mallows noise parameter; we can then directly calculate $t_{ij}^i(K), t_{ij}^j(K)$ and thus the learning rate for each $K$-Approval mechanism. This program leverages Mallows repeated insertion probabilities \cite{lu14a,diaconis_group_nodate} and may be of independent interest for numerical analyses of the Mallows model.

\parbold{Numerical analysis} We now numerically analyze, for the Mallows model, how the Approval rate approval $K$-Approval mechanism varies with the Mallows noise parameter $\phi$, the number of candidates $M$, and the number of winners $W$. Recall that the Mallows model is asymptotically design invariant, so different mechanisms only differ in \textit{how quickly} they learn the asymptotic outcome.

In Figure~\ref{fig:mallowsnumcandid}, the goal is to select $W=1$ winner, and $\phi$ and $M$ are varied. With low noise, $\phi \lessapprox  .5$, it is rate optimal to use $1$-Approval, i.e., ask each voter to select their favorite candidate, regardless of how many candidates there are. However, with higher noise $\phi$, as the number of candidates in the election increases, so does the $K$ in the optimal $K$-Approval mechanism. For $\phi = .999, M = 50$, for example, it is best to ask each voter to select their favorite $25$ candidates, even if the task is to identify the single best candidate according to the reference distribution $\sigma^*$.

Similarly, Figure~\ref{fig:mallowsnumwinners} shows how the rate optimal $K$-Approval mechanism changes with the number of winners desired and the noise parameter, fixing the number of candidates at $M=50$. Again with high noise, it is best to ask voters to identify their favorite half of candidates, regardless of how many winners need to be identified. With low noise, however, $W$-Approval is optimal to select $W$ winners. 

Overall, the analysis suggests that with higher noise in the voter model, one should tend toward asking voters to rank their favorite half of candidates, regardless of $M$ and $W$. 

The high-noise setting may seem unrealistic; however, as we will see in the next section, which $K$-Approval mechanism is rate optimal in practice often scales like the high noise settings, consistent with the idea that voting distributions in practice do not look like they are drawn from a low-noise Mallows model. We now turn to such empirical analyses.

%
%
%
%
%
%
%
%
%
%
%

%

%

%
%

%
%

%
%

%

%

%
%
\makeatletter{}%
\section{Empirics and PB deployments}
\label{sec:data}

\begin{figure*}
	\centering
	\begin{subfigure}[b]{.48\textwidth}
		\includegraphics[width=\linewidth]{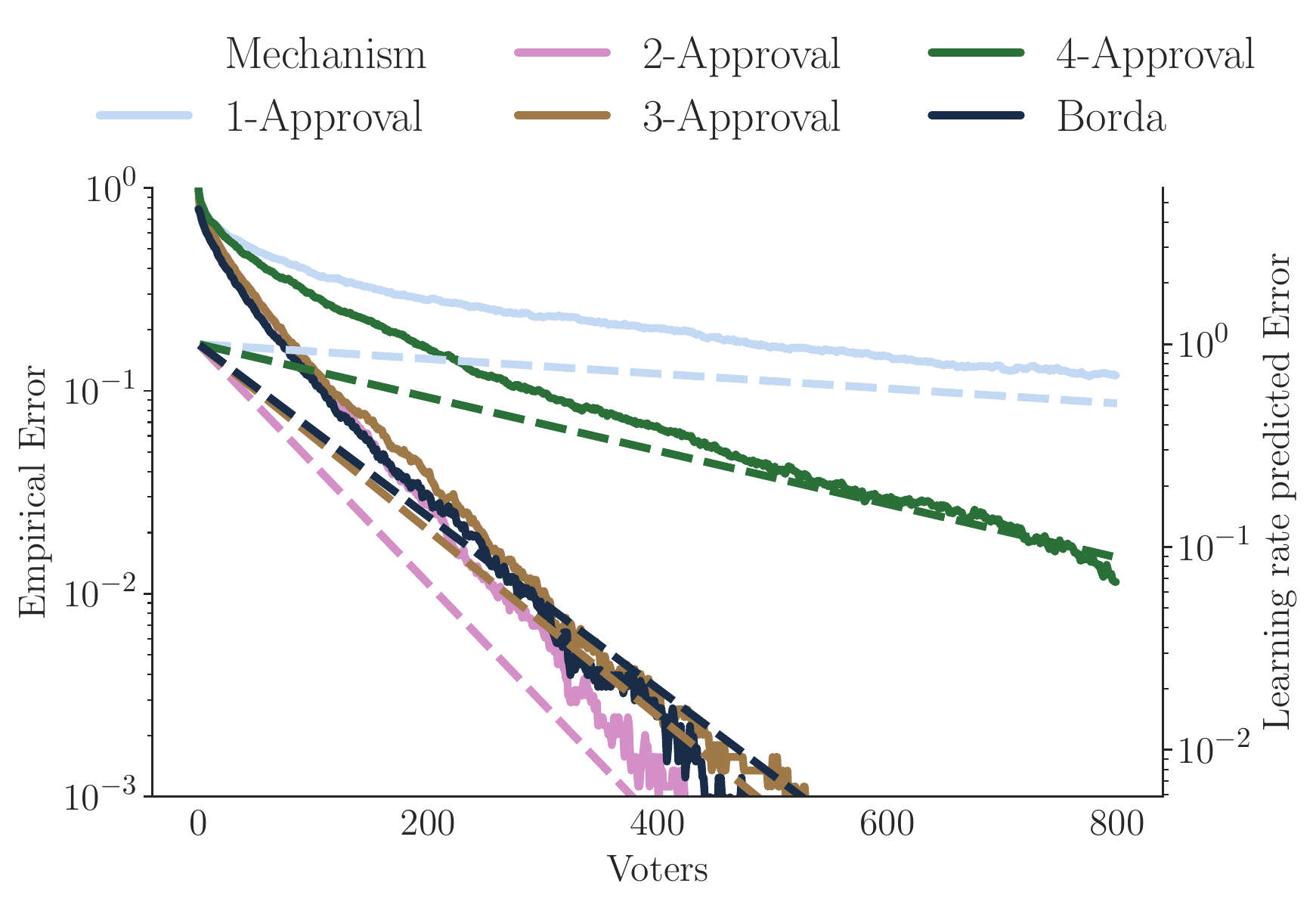}\vfill
		\caption{Boston 2016 PB election, selecting 1 winner: Average empirical bootstrapped error -- i.e., fraction of times the asymptotic winner is selected (solid lines, left axis), compared to such errors over time implied by the (empirically calculated) learning rates -- i.e., $e^{-rN}$ (dashed lines, right axis). The right axis is a vertically shifted (in log scale) version of the left axis, reflecting that the learning rate errors are asymptotically valid up to polynomial factors. All mechanisms return the same winner when all votes are counted. ``Borda'' is the Borda count for the 4 candidates ranked, and all others are assumed to be tied at rank 5 for each voter.}
	\label{fig:errorovertime}
	\end{subfigure}\hfill
	\begin{subfigure}[b]{.48\textwidth}
		\includegraphics[width=\linewidth]{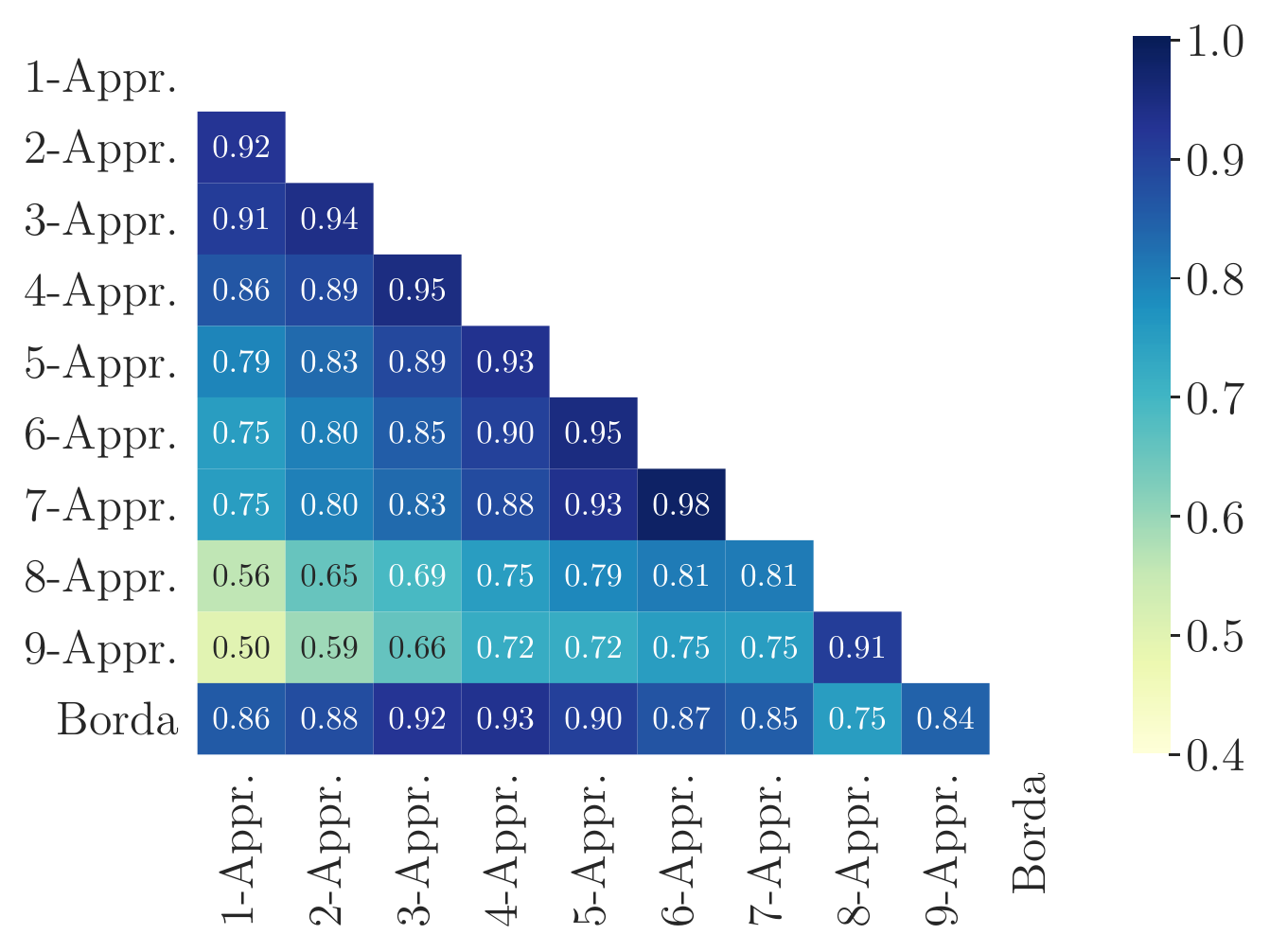}\vfill
		\caption{Approximate design invariance across elections. For the task of selecting $W=4$ winners, this plot shows the average overlap in the top $4$ candidates identified by different mechanisms across all the elections in our dataset, if all voters with complete rankings are counted. For example, of the top 4 candidates identified by $1$-Approval across elections, $92\%$ are also identified as top $4$ candidates by $2$-Approval. For each $K$-Approval mechanism, we include all elections where there were at least $K+1$ candidates. }
		\label{fig:approxdesigninvar}
	\end{subfigure}\hfill
	\caption{ Validating model: comparing learning rates to empirical error, and showing approximate design invariance.}

\end{figure*}

We now apply our insights to practice. We focus on $K$-Approval voting, as opposed to general scoring rules. This section is organized as follows. In Section~\ref{sec:datadesc}, we describe our data sources. We validate our model in Section~\ref{sec:modelvalidation}; first, we demonstrate that large deviation rates effectively describe how quickly various mechanisms learn; next, we show that in practice voter noise models are approximately design invariant. In Section~\ref{sec:approvaldata}, we show that the insights from Section~\ref{sec:approvaltheory} regarding optimal approval mechanisms extend to practice. Finally in Section~\ref{sec:randomizationpractice} we note that we find many examples in practice where randomizing between $K$-Approval mechanisms improves learning.

\subsection{Data description}
\label{sec:datadesc}
 We leverage two data sources (detailed dataset information is in Appendix Table~\ref{tab:datasources}). First, we have partnered with dozens of local governments to help run participatory budgeting (PB) elections in the last five years. These elections have used a variety of methods, primarily $K$-Approval; our data in this work comes from 5 elections where $K$-Ranking was used, including 3 recent elections where $K=10$. This data is particularly useful as PB is among the most common types of elections with many candidates and several winners, with several theoretical analyses~\cite{goel_knapsack_2016,garg2019iterative,freeman2019truthful}. %

  Second, we use data available on PrefLib~\cite{MaWa13a,preflibstv,regenwetter2007unexpected,regenwetter2008behavioural,popov2014consensus}, limiting ourselves to 28 elections with at least 5 candidates and 700 voters who provided full rankings. This ranking data spans many domains, from people's sushi preferences to Glasgow City Council elections. This domain breadth supports the broad applicability of the design insights explored in this section.

We focus on ranking data to be able to simulate counter-factuals for the same election: with $K$-Ranking data, we can simulate what would have occurred with any $K'$-Approval elicitation mechanism, for $K'\leq K$ (assuming no behavioral quirks). With approval data, on the other hand, one cannot compare the mechanism to any other for that given election.

One challenge is that ranking many candidates is onerous, and so voters rank at most 14 candidates in our dataset. For the data we use from on PrefLib, full rankings (rankings up to the number of candidates) are available. In the PB elections in our partner cities, typically each voter ranks or selects her favorite $K \ll M$ candidates.

\subsection{Model validation}
\label{sec:modelvalidation}
Our model and design approach has two components that must be validated: (1) that learning rates can effectively be used to compare different mechanisms, and (2) that design invariance (approximately) holds in practice.

\parbold{Large deviation rates as effective proxies for learning}
We now confirm that, for a given election, empirically calculated large deviation learning rates are effective proxies for the rate at which the error in recovering the asymptotic output decreases as the number of voters increases (even though large deviation learning rates are only asymptotically valid in the number of voters). As examples, we first identify three elections and goals for which many of the potential $K$-Approval mechanisms return exactly the same asymptotic outcome. Then, we bootstrap voters from the available data of voters and empirically calculate the errors made in identifying the winning set of candidates. We further calculate the large deviation learning rates for these mechanisms, using $F$ implied by the voting data and the formula in Proposition~\ref{thm:goal_learning}.\footnote{Given an empirical $\hat F$, learning rates can be numerically calculated: the $\inf_z[\cdot]$ is a convex minimization problem.}

Figure~\ref{fig:errorovertime} shows the resulting errors over time for one such election where $4$-Rankings are available. We further plot  $e^{-rN}$ for each mechanism, i.e., the error over time implied by the learning rate (up to polynomial factors). This plot, along with Appendix Figure~\ref{fig:errorovertime_app}, yields several insights: 
\begin{enumerate}[(1)]
		\item The mechanism matters: when selecting 1 winner from the election in Figure~\ref{fig:errorovertime} after $400$ votes, there is $20\%$ chance of not picking the ultimate winner if $1$-Approval is used. With $2$ or $3$-Approval, this number is $0.1\%$. The winner appears often in a voter's top two or three positions (but not necessarily first), while the ultimate second place candidate often falls outside the top three. Scoring rules that reward top three placements thus perform well. 
	\item The learning rates effectively capture the behavior of the empirical error: both comparatively across mechanisms, as well as the asymptotic rate (slope of the line in log scale). This property enables use of large deviation learning rates as proxies for learning even in  elections with a small number of voters.
	\item Ranking $K$ candidates rather than selecting $K$ candidates is more onerous for voters. However, it does not always provide more information in terms of learning rates, as in the examples in Appendix Figure~\ref{fig:errorovertime_app}.%

\end{enumerate}

\parbold{Design invariance in practice} Design invariance does not strictly hold  in any election in our dataset (as expected as the condition is strong). However, it \textit{approximately} holds. Similar mechanisms produce the same asymptotic outcome for many tasks. Figure~\ref{fig:approxdesigninvar} shows, for example, the average overlap across elections in the top $4$ candidates identified by each mechanism. (Appendix Figure~\ref{fig:approxdesigninvar_app} shows the same plot for the top $1$ and $3$ candidates, as well as the average Kendall's $\tau$ rank correlation between the full rankings identified by different mechanisms). Furthermore, we find many elections and goals where most mechanisms return the same asymptotic answer, as in the elections we leverage for the plots showing learning rates are effective proxies. This relative consistency, especially for similar mechanisms, enables us to compare different mechanisms by their learning rates.

\subsection{$K$-Approval for selecting $W$ winners}
\label{sec:approvaldata}
 In Section~\ref{sec:approvaltheory}, we showed for the Mallows model how the rate optimal $K$-Approval mechanism changes with the noise parameter $\phi$, the number of candidates, and the number of winners. We now show this scaling in practice. 

For every election in our dataset, we find the Approval rate optimal mechanism (among $K$ we can simulate) for every goal of selecting $W$ winners, for $1 \leq W \leq M$. We then run a regression across all the elections for which $K$ is rate optimal, versus the number of winners desired and the number of candidates; see Table~\ref{tab:kregression} in the Appendix for the regression table. While there is some variation across elections, the number of candidates and winners proves a reasonable metric across elections for the rate approval $K$-Approval mechanism ($R^2 \approx .27$). 

The regression confirms the idea that in practice, one should regularize toward asking voters to choose their favorite half the candidates. For picking a small subset of winners $W \approx 4$ out of more than $10$ candidates, for example, one should ask voters to provide their favorite $K \approx 6$ candidates, with $K> W$. This suggestion directly counters common practice. In the PB elections that we have helped run, for example, $4$ or $5$-Approval is most typical, even though ultimately 6-10 projects may be funded (out of $\approx15$-$20$). 

Then, in Figure~\ref{fig:mallowsplotsapp} in the Appendix, we plot the line induced from the regression coefficients with the Mallows rate optimal lines, for $M\leq 10$ candidates.  Comparing to the rate optimal mechanisms for the Mallows model with various $\phi$ (within the candidate range for which we have empirical data), we find that empirical data behaves most closely to a Mallows model with noise parameter $\phi \in [.8, .9]$. (We are not claiming that empirical data is drawn from a Mallows model; it most certainly is not, with factors such as polarizing projects important in practice). This coarse comparison provides an approximate expected scaling behavior for elections with many candidates. 

\subsection{Randomization in practice}
\label{sec:randomizationpractice}

We find 16 examples in which randomizing between two $K$-Approval mechanisms leads to faster learning than using either mechanism separately, including 8 examples where such randomization beats the Approval rate optimal mechanism. Table~\ref{tab:randomization} in the Appendix contains details. 

%

%

%
%

%

%

%
%
%
%
%
%
%

%
%
%
%

%
\makeatletter{}%
\section{Discussion}
We show that in elections with many candidates, the elicitation mechanism and corresponding scoring rule used affect how quickly the final outcome is learned. The learning speed differential between mechanisms can be the difference between identifying the ultimate winner with only a $80\%$ probability or a $99.9\%$ probability after 400 voters, for example. We then provide design decisions that emerge when our framework is applied to data from real elections. When using $K$-Approval to select a small number of $W$ winners, for example, it is often better to ask voters to identify their favorite $K>W$ candidates. The insights from this work should be applicable in a variety of such settings, from elections to crowdsourcing labeling tasks.

There are several important, open research avenues. Most importantly, in real elections maximizing the rate at which the final outcome is identified is not the only goal, and future work should seek to balance such multiple objectives.

For example, there may be axiomatic reasons to prefer one elicitation mechanism over another, e.g., that the final outcome corresponds to the candidate(s) that the most voters indicate is their first choice. Another objective may be to minimize the cognitive load imposed on voters. Asking voters to provide a full ranking over the candidates and then using a rate-optimal scoring rule trivially provides faster learning than any other mechanism. However, asking voters to rank 20 candidates is prohibitive in many settings. Future empirical work, in line with that of \citet{benade_efciency_2018}~and~\citet{gelauff_comparing_2018}, should study the cognitive load various mechanisms impose on voters, to better understand the trade-off between the objectives.    %

\section*{Acknowledgments}
	
We thank our Participatory Budgeting city partners, especially those in Boston, Durham, and Rochester. We also thank anonymous reviewers for their comments. This work was supported in part by the Stanford Cyber Initiative, the Office of Naval Research grant N00014-19-1-2268, and National Science Foundation grants 1544548 and 1637397.

\bibliographystyle{aaai}
\bibliography{bib}

\begin{thebibliography}{}

\bibitem[\protect\citeauthoryear{Al{\'o}s-Ferrer and
  Grani{\'c}}{2012}]{alos-ferrer_two_2011}
Al{\'o}s-Ferrer, C., and Grani{\'c}, D.-G.
\newblock 2012.
\newblock Two field experiments on approval voting in germany.
\newblock {\em Social Choice and Welfare} 39(1):171--205.

\bibitem[\protect\citeauthoryear{Aziz \bgroup et al\mbox.\egroup
  }{2015}]{aziz_computational_2015}
Aziz, H.; Gaspers, S.; Gudmundsson, J.; Mackenzie, S.; Mattei, N.; and Walsh,
  T.
\newblock 2015.
\newblock Computational aspects of multi-winner approval voting.
\newblock In {\em Proceedings of the 2015 International Conference on
  Autonomous Agents and Multiagent Systems},  107--115.
\newblock International Foundation for Autonomous Agents and Multiagent
  Systems.

\bibitem[\protect\citeauthoryear{Aziz \bgroup et al\mbox.\egroup
  }{2017}]{aziz_justified_2017}
Aziz, H.; Brill, M.; Conitzer, V.; Elkind, E.; Freeman, R.; and Walsh, T.
\newblock 2017.
\newblock Justified representation in approval-based committee voting.
\newblock {\em Social Choice and Welfare} 48(2):461--485.

\bibitem[\protect\citeauthoryear{Benade \bgroup et al\mbox.\egroup
  }{2018}]{benade_efciency_2018}
Benade, G.; Itzhak, N.; Shah, N.; and Procaccia, A.~D.
\newblock 2018.
\newblock Efﬁciency and {Usability} of {Participatory} {Budgeting} {Methods}.
\newblock ~8.

\bibitem[\protect\citeauthoryear{Boyd and
  Vandenberghe}{2004}]{boyd_convex_2004}
Boyd, S.~P., and Vandenberghe, L.
\newblock 2004.
\newblock {\em Convex optimization}.
\newblock Cambridge, UK ; New York: Cambridge University Press.

\bibitem[\protect\citeauthoryear{Caragiannis and
  Micha}{2017}]{caragiannis_learning_2017}
Caragiannis, I., and Micha, E.
\newblock 2017.
\newblock Learning a {Ground} {Truth} {Ranking} {Using} {Noisy} {Approval}
  {Votes}.
\newblock In {\em Proceedings of the {Twenty}-{Sixth} {International} {Joint}
  {Conference} on {Artificial} {Intelligence}},  149--155.
\newblock Melbourne, Australia: International Joint Conferences on Artificial
  Intelligence Organization.

\bibitem[\protect\citeauthoryear{Caragiannis \bgroup et al\mbox.\egroup
  }{2017}]{caragiannis_subset_2017}
Caragiannis, I.; Nath, S.; Procaccia, A.~D.; and Shah, N.
\newblock 2017.
\newblock Subset selection via implicit utilitarian voting.
\newblock {\em Journal of Artificial Intelligence Research} 58:123--152.

\bibitem[\protect\citeauthoryear{Caragiannis \bgroup et al\mbox.\egroup
  }{2019}]{caragiannis2019optimizing}
Caragiannis, I.; Chatzigeorgiou, X.; Krimpas, G.~A.; and Voudouris, A.~A.
\newblock 2019.
\newblock Optimizing positional scoring rules for rank aggregation.
\newblock {\em Artificial Intelligence} 267:58--77.

\bibitem[\protect\citeauthoryear{Caragiannis, Procaccia, and
  Shah}{2013}]{caragiannis_when_2013}
Caragiannis, I.; Procaccia, A.~D.; and Shah, N.
\newblock 2013.
\newblock When do noisy votes reveal the truth?
\newblock In {\em Proceedings of the fourteenth ACM conference on Electronic
  commerce},  143--160.
\newblock ACM.

\bibitem[\protect\citeauthoryear{Chierichetti and
  Kleinberg}{2014}]{chierichetti2014voting}
Chierichetti, F., and Kleinberg, J.
\newblock 2014.
\newblock Voting with limited information and many alternatives.
\newblock {\em SIAM Journal on Computing} 43(5):1615--1653.

\bibitem[\protect\citeauthoryear{Copeland}{1951}]{copeland1951reasonable}
Copeland, A.~H.
\newblock 1951.
\newblock A reasonable social welfare function.
\newblock Technical report, mimeo, 1951. University of Michigan.

\bibitem[\protect\citeauthoryear{de Borda}{1781}]{de1781memoire}
de~Borda, J.~C.
\newblock 1781.
\newblock M{\'e}moire sur les {\'e}lections au scrutin.

\bibitem[\protect\citeauthoryear{de Weerdt, Gerding, and
  Stein}{2016}]{de2016minimising}
de~Weerdt, M.~M.; Gerding, E.~H.; and Stein, S.
\newblock 2016.
\newblock Minimising the rank aggregation error.
\newblock In {\em Proceedings of the 2016 International Conference on
  Autonomous Agents \& Multiagent Systems},  1375--1376.
\newblock International Foundation for Autonomous Agents and Multiagent
  Systems.

\bibitem[\protect\citeauthoryear{Dembo and Zeitouni}{2010}]{dembo_large_2010}
Dembo, A., and Zeitouni, O.
\newblock 2010.
\newblock {\em Large {Deviations} {Techniques} and {Applications}}, volume~38
  of {\em Stochastic {Modelling} and {Applied} {Probability}}.
\newblock Berlin, Heidelberg: Springer Berlin Heidelberg.

\bibitem[\protect\citeauthoryear{Diaconis}{1988}]{diaconis_group_nodate}
Diaconis, P.
\newblock 1988.
\newblock Group representations in probability and statistics.
\newblock {\em Lecture notes-monograph series} 11:i--192.

\bibitem[\protect\citeauthoryear{Elkind \bgroup et al\mbox.\egroup
  }{2017}]{elkind_properties_2015}
Elkind, E.; Faliszewski, P.; Skowron, P.; and Slinko, A.
\newblock 2017.
\newblock Properties of {Multiwinner} {Voting} {Rules}.
\newblock {\em Social Choice and Welfare} 48(3):599--632.

\bibitem[\protect\citeauthoryear{Faliszewski and
  Talmon}{2018}]{faliszewski_framework_2018}
Faliszewski, P., and Talmon, N.
\newblock 2018.
\newblock A framework for approval-based budgeting methods.
\newblock {\em arXiv preprint arXiv:1809.04382}.

\bibitem[\protect\citeauthoryear{Fishburn and
  Gehrlein}{1976}]{fishburn_bordas_1976}
Fishburn, P.~C., and Gehrlein, W.~V.
\newblock 1976.
\newblock Borda's rule, positional voting, and {Condorcet}'s simple majority
  principle.
\newblock {\em Public Choice} 28(1):79--88.

\bibitem[\protect\citeauthoryear{Fishburn}{1978}]{fishburn_axioms_1978}
Fishburn, P.~C.
\newblock 1978.
\newblock Axioms for approval voting: Direct proof.
\newblock {\em Journal of Economic Theory} 19(1):180--185.

\bibitem[\protect\citeauthoryear{Freeman \bgroup et al\mbox.\egroup
  }{2019}]{freeman2019truthful}
Freeman, R.; Pennock, D.~M.; Peters, D.; and Vaughan, J.~W.
\newblock 2019.
\newblock Truthful aggregation of budget proposals.
\newblock {\em arXiv preprint arXiv:1905.00457}.

\bibitem[\protect\citeauthoryear{Garg and Johari}{2018}]{garg2018designing}
Garg, N., and Johari, R.
\newblock 2018.
\newblock Designing informative rating systems: Evidence from an online labor
  market.
\newblock {\em arXiv preprint arXiv:1810.13028}.

\bibitem[\protect\citeauthoryear{Garg and Johari}{2019}]{garg2019designing}
Garg, N., and Johari, R.
\newblock 2019.
\newblock Designing optimal binary rating systems.
\newblock In {\em Proceedings of the 22nd International Conference on
  Artificial Intelligence and Statistics}.

\bibitem[\protect\citeauthoryear{Garg \bgroup et al\mbox.\egroup
  }{2019}]{garg2019iterative}
Garg, N.; Kamble, V.; Goel, A.; Marn, D.; and Munagala, K.
\newblock 2019.
\newblock Iterative local voting for collective decision-making in continuous
  spaces.
\newblock {\em Journal of Artificial Intelligence Research} 64(1):315--355.

\bibitem[\protect\citeauthoryear{Gelauff \bgroup et al\mbox.\egroup
  }{2018}]{gelauff_comparing_2018}
Gelauff, L.; Sakshuwong, S.; Garg, N.; and Goel, A.
\newblock 2018.
\newblock Comparing voting methods for budget decisions on the {ASSU} ballot.
\newblock Technical report.

\bibitem[\protect\citeauthoryear{Goel \bgroup et al\mbox.\egroup
  }{2016}]{goel_knapsack_2016}
Goel, A.; Krishnaswamy, A.~K.; Sakshuwong, S.; and Aitamurto, T.
\newblock 2016.
\newblock Knapsack {Voting}: {Voting} mechanisms for {Participatory}
  {Budgeting}.

\bibitem[\protect\citeauthoryear{Guiver and
  Snelson}{2009}]{guiver_bayesian_2009}
Guiver, J., and Snelson, E.
\newblock 2009.
\newblock Bayesian inference for plackett-luce ranking models.
\newblock In {\em proceedings of the 26th annual international conference on
  machine learning},  377--384.
\newblock ACM.

\bibitem[\protect\citeauthoryear{Kemeny}{1959}]{kemeny1959mathematics}
Kemeny, J.~G.
\newblock 1959.
\newblock Mathematics without numbers.
\newblock {\em Daedalus} 88(4):577--591.

\bibitem[\protect\citeauthoryear{Lackner and
  Skowron}{2018a}]{lackner_consistent_2018}
Lackner, M., and Skowron, P.
\newblock 2018a.
\newblock Consistent {Approval}-{Based} {Multi}-{Winner} {Rules}.
\newblock In {\em Proceedings of the 2018 {ACM} {Conference} on {Economics} and
  {Computation}}, {EC} '18,  47--48.
\newblock New York, NY, USA: ACM.

\bibitem[\protect\citeauthoryear{Lackner and
  Skowron}{2018b}]{lackner_quantitative_2018}
Lackner, M., and Skowron, P.
\newblock 2018b.
\newblock A quantitative analysis of multi-winner rules.
\newblock {\em arXiv preprint arXiv:1801.01527}.

\bibitem[\protect\citeauthoryear{Lee \bgroup et al\mbox.\egroup
  }{2014}]{lee2014crowdsourcing}
Lee, D.~T.; Goel, A.; Aitamurto, T.; and Landemore, H.
\newblock 2014.
\newblock Crowdsourcing for participatory democracies: Efficient elicitation of
  social choice functions.
\newblock In {\em Second AAAI Conference on Human Computation and
  Crowdsourcing}.

\bibitem[\protect\citeauthoryear{Lu and Boutilier}{2011}]{lu_learning_2011}
Lu, T., and Boutilier, C.
\newblock 2011.
\newblock Learning {Mallows} {Models} with {Pairwise} {Preferences}.
\newblock In {\em Proceedings of the 28th {International} {Conference} on
  {International} {Conference} on {Machine} {Learning}}, {ICML}'11,  145--152.
\newblock USA: Omnipress.

\bibitem[\protect\citeauthoryear{Lu and Boutilier}{2014}]{lu14a}
Lu, T., and Boutilier, C.
\newblock 2014.
\newblock Effective sampling and learning for mallows models with
  pairwise-preference data.
\newblock {\em Journal of Machine Learning Research} 15:3963--4009.

\bibitem[\protect\citeauthoryear{Mallows}{1957}]{mallows1957non}
Mallows, C.~L.
\newblock 1957.
\newblock Non-null ranking models. i.
\newblock {\em Biometrika} 44(1/2):114--130.

\bibitem[\protect\citeauthoryear{marquis~de Condorcet}{1785}]{marquis1785essai}
marquis~de Condorcet, M. J.~A.
\newblock 1785.
\newblock {\em Essai sur l'application de l'analyse a la probabilite des
  decisions: rendues a la pluralite de voix}.
\newblock De l'Imprimerie royale.

\bibitem[\protect\citeauthoryear{Mattei and Walsh}{2013}]{MaWa13a}
Mattei, N., and Walsh, T.
\newblock 2013.
\newblock Preflib: A library of preference data \textsc{http://preflib.org}.
\newblock In {\em Proceedings of the 3rd International Conference on
  Algorithmic Decision Theory (ADT 2013)}, Lecture Notes in Artificial
  Intelligence.
\newblock Springer.

\bibitem[\protect\citeauthoryear{Maystre and
  Grossglauser}{2015}]{maystre_fast_2015}
Maystre, L., and Grossglauser, M.
\newblock 2015.
\newblock Fast and accurate inference of plackett--luce models.
\newblock In {\em Advances in neural information processing systems}.
\newblock  172--180.

\bibitem[\protect\citeauthoryear{O'Neill}{2013}]{preflibstv}
O'Neill, J.
\newblock 2013.
\newblock Open {S}{T}{V}.

\bibitem[\protect\citeauthoryear{Popov, Popova, and
  Regenwetter}{2014}]{popov2014consensus}
Popov, S.~V.; Popova, A.; and Regenwetter, M.
\newblock 2014.
\newblock Consensus in organizations: Hunting for the social choice conundrum
  in apa elections.
\newblock {\em Decision} 1(2):123.

\bibitem[\protect\citeauthoryear{Procaccia and Shah}{2015}]{procaccia_is_2015}
Procaccia, A.~D., and Shah, N.
\newblock 2015.
\newblock Is {Approval} {Voting} {Optimal} {Given} {Approval} {Votes}?
\newblock In {\em Advances in {Neural} {Information} {Processing} {Systems}
  28},  1801--1809.

\bibitem[\protect\citeauthoryear{{Public
  Agenda}}{2016}]{public_agenda_public_2016}
{Public Agenda}.
\newblock 2016.
\newblock Public {Spending} {By} {The} {People}: {Participatory} {Budgeting} in
  the {United} {States} and {Canada} in 2014–15.
\newblock Technical report, The Yankelovich Center for Public Judgment.

\bibitem[\protect\citeauthoryear{Ratliff}{2003}]{ratliff_startling_2003}
Ratliff, T.~C.
\newblock 2003.
\newblock Some startling inconsistencies when electing committees.
\newblock {\em Social Choice and Welfare} 21(3):433--454.

\bibitem[\protect\citeauthoryear{Regenwetter \bgroup et al\mbox.\egroup
  }{2007}]{regenwetter2007unexpected}
Regenwetter, M.; Kim, A.; Kantor, A.; and Ho, M.-H.~R.
\newblock 2007.
\newblock The unexpected empirical consensus among consensus methods.
\newblock {\em Psychological Science} 18(7):629--635.

\bibitem[\protect\citeauthoryear{Regenwetter \bgroup et al\mbox.\egroup
  }{2008}]{regenwetter2008behavioural}
Regenwetter, M.; Grofman, B.; Popova, A.; Messner, W.; Davis-Stober, C.~P.; and
  Cavagnaro, D.~R.
\newblock 2008.
\newblock Behavioural social choice: a status report.
\newblock {\em Philosophical Transactions of the Royal Society B: Biological
  Sciences} 364(1518):833--843.

\bibitem[\protect\citeauthoryear{Staring}{1986}]{staring_two_1986}
Staring, M.
\newblock 1986.
\newblock Two paradoxes of committee elections.
\newblock {\em Mathematics Magazine} 59(3):158--159.

\bibitem[\protect\citeauthoryear{Tataru and
  Merlin}{1997}]{tataru_relationship_1997}
Tataru, M., and Merlin, V.
\newblock 1997.
\newblock On the relationship of the {Condorcet} winner and positional voting
  rules.
\newblock {\em Mathematical Social Sciences} 34(1):81--90.

\bibitem[\protect\citeauthoryear{Wiseman}{2000}]{wiseman_approval_2000}
Wiseman, J.
\newblock 2000.
\newblock Approval voting in subset elections.
\newblock {\em Economic Theory} 15(2):477--483.

\bibitem[\protect\citeauthoryear{Young}{1975}]{young1975social}
Young, H.~P.
\newblock 1975.
\newblock Social choice scoring functions.
\newblock {\em SIAM Journal on Applied Mathematics} 28(4):824--838.

\bibitem[\protect\citeauthoryear{Young}{1988}]{young1988condorcet}
Young, H.~P.
\newblock 1988.
\newblock Condorcet's theory of voting.
\newblock {\em American Political science review} 82(4):1231--1244.

\bibitem[\protect\citeauthoryear{Zhao, Piech, and
  Xia}{2016}]{zhao_learning_2016}
Zhao, Z.; Piech, P.; and Xia, L.
\newblock 2016.
\newblock Learning {Mixtures} of {Plackett}-{Luce} {Models}.
\newblock In {\em International Conference on Machine Learning},  2906--2914.

\end{thebibliography}

\onecolumn
\appendix
\makeatletter{}%
\section{Empirics additional information}
\FloatBarrier
\begin{table}[h]
	\small
	\begin{center}
\begin{tabular}{lccc}
	\textbf{Name} &  \textbf{Candidates} &  \textbf{Votes with complete rankings} &  \textbf{K-Ranking available}  \\&&&\\
	\textbf{Participatory Budgeting}&&&\\
	Boston, 2016 &           8 &   4173 &                    4\\
	Durham Ward 1, 2019 &          21 &    1637 &                   10\\
	Durham Ward 2, 2019 &          10 &    329 &                   10 \\
	Durham Ward 3, 2019 &          12 &    694 &                   10    \\
	Rochester, 2019 &          22 &    649 &                    5\\&&&\\\textbf{PrefLib:}&&&\\
	Irish01 &          12 &   4259 &                   12 \\
	Irish02 &           9 &   4810 &                    9 \\
	Irish03 &          14 &   3166 &                   14 \\
	ElectorialReformSociety77 &          12 &   1312 &                   12 \\
	ElectorialReformSociety13 &           5 &   1809 &                    5 \\
	Sushi10 &          10 &   5000 &                   10 \\
	Glasgow05 &          10 &    718 &                   10  \\
	Glasgow17 &           9 &    962 &                    9  \\
	Glasgow10 &           9 &    818 &                    9  \\
	Glasgow18 &           9 &    767 &                    9  \\
	Glasgow20 &           9 &    726 &                    9  \\
	Glasgow14 &           8 &   1071 &                    8  \\
	Glasgow12 &           8 &   1040 &                    8  \\
	Burlington01 &           6 &   2603 &                    6  \\
	Burlington02 &           6 &   2853 &                    6  \\
	APA03 &           5 &  11539 &                    5 \\
	APA01 &           5 &  10978 &                    5 \\
	APA11 &           5 &  10791 &                    5 \\
	APA05 &           5 &  10655 &                    5 \\
	APA02 &           5 &  10623 &                    5 \\
	APA04 &           5 &  10519 &                    5 \\
	APA09 &           5 &  10211 &                    5 \\
	APA06 &           5 &  10177 &                    5 \\
	APA07 &           5 &   9747 &                    5 \\
	APA12 &           5 &   9091 &                    5 \\
	APA08 &           5 &   8532 &                    5 \\
	APA10 &           5 &   8467 &                    5 \\
	Aspen02 &           5 &   1183 &                    5 
\end{tabular}
\caption{List of election data that we use in Section~\ref{sec:data}. From PrefLib, we use all elections where full rankings are available and there are at least 5 candidates and 700 voters. Throughout, we ignore voters who did not submit full rankings (especially with high $K$-Ranking requested, this might only be a fraction of the total number of actual votes). Additionally, for the PB elections, we limit the data to those who submitted votes online rather than through paper ballots. 
	~\\\\
	Sources for the PrefLib datasets are: \cite{MaWa13a,preflibstv,regenwetter2007unexpected,regenwetter2008behavioural,popov2014consensus}.}
\label{tab:datasources}
\end{center}
\end{table}

\begin{figure*}
	\centering
	\begin{subfigure}[b]{.48\textwidth}
		\includegraphics[width=\linewidth]{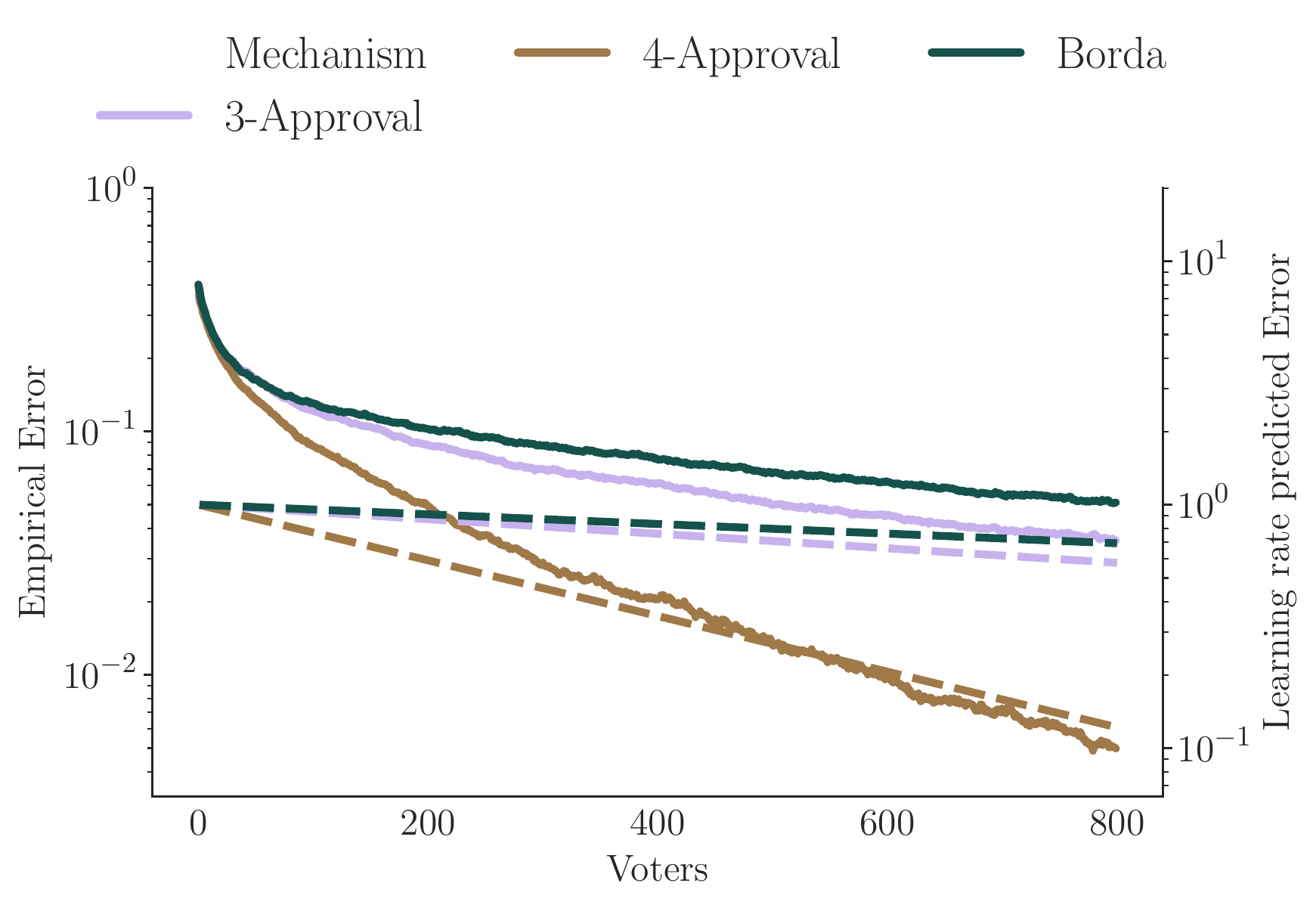}\vfill
		\caption{Boston 2016, selecting 4 winners.}
		\label{fig:boston164}
	\end{subfigure}\hfill
	\begin{subfigure}[b]{.48\textwidth}
		\includegraphics[width=\linewidth]{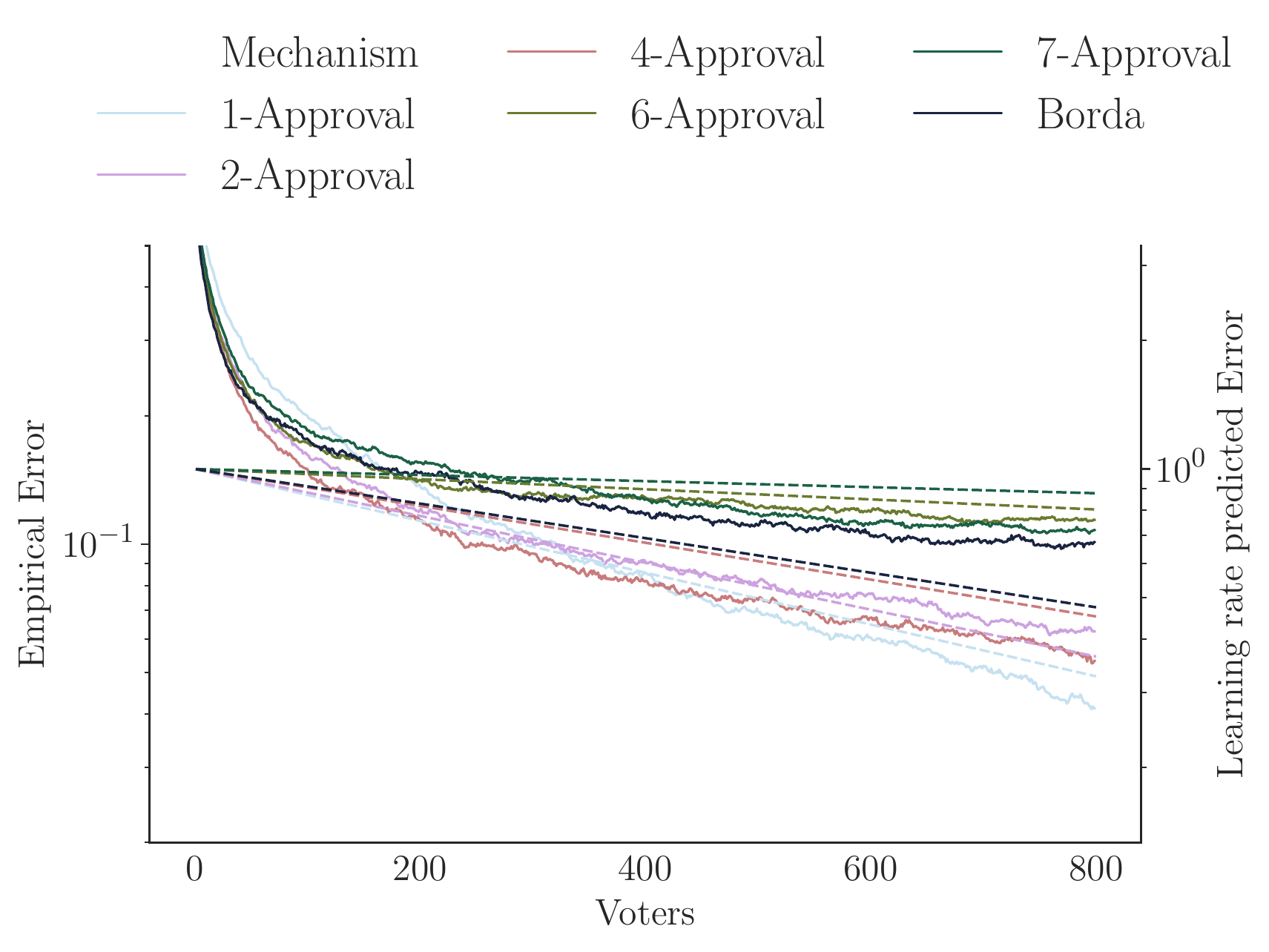}\vfill
		\caption{Durham Ward 1, selecting 4 winners. $K$-Approval for $K\in\{1, \dots, 7\}$ and the Borda rule all have the same asymptotic winners, but we omit several mechanisms from the plot for visualization ease. %
}
		\label{fig:durham191_7}

	\end{subfigure}\hfill
	\caption{Average bootstrapped error (fraction of winning subset not identified) by the number of voters, compared to the errors implied by the (empirically calculated) learning rates. All mechanisms plotted have the same asymptotic winners.}
	\label{fig:errorovertime_app}
\end{figure*}

\begin{figure*}
	\centering
	\begin{subfigure}[b]{.48\textwidth}
	\includegraphics[width=\linewidth]{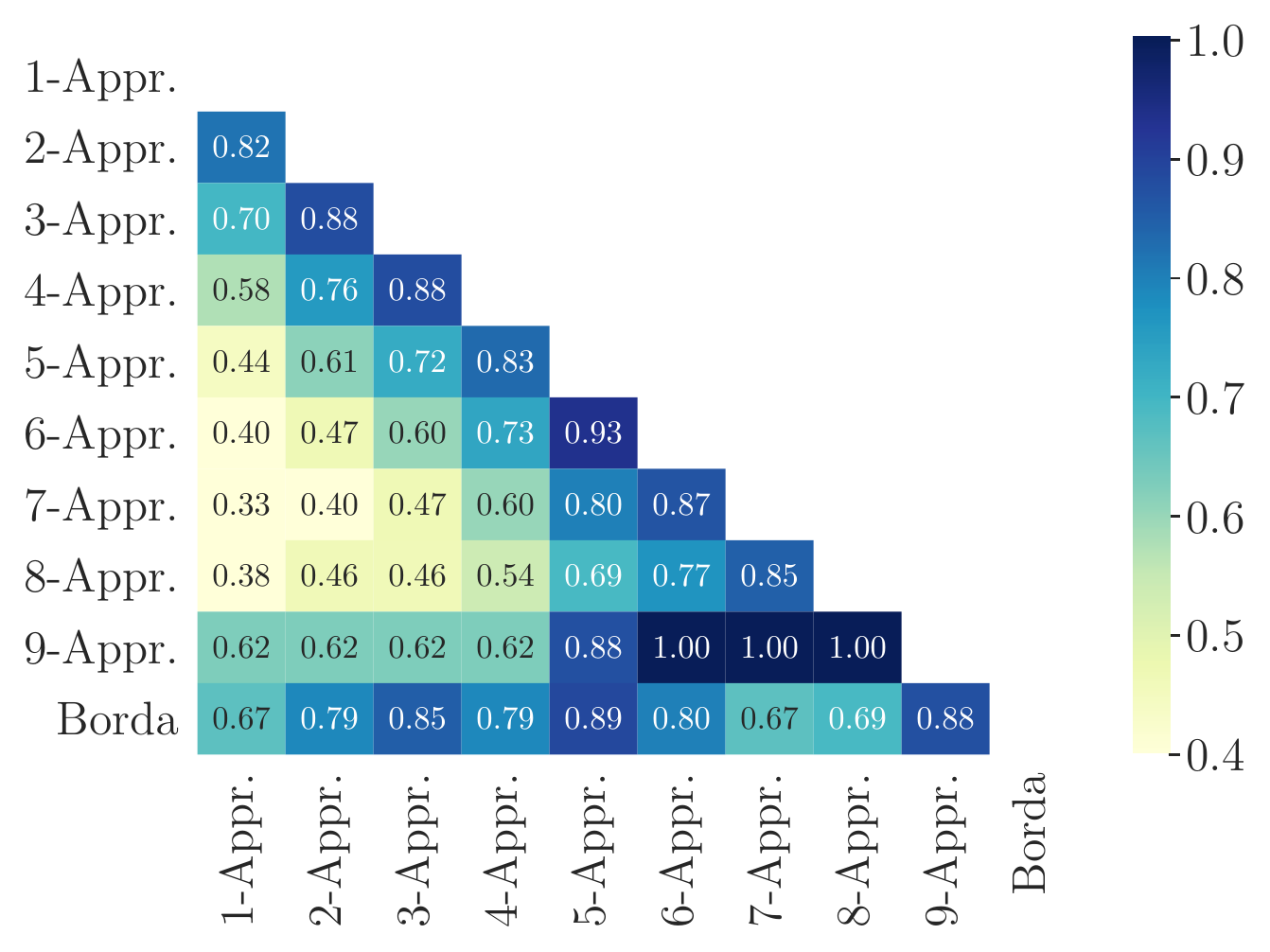}\vfill
	\caption{Task of selecting $W=1$ winners. }
\end{subfigure}\hfill
	\begin{subfigure}[b]{.48\textwidth}
		\includegraphics[width=\linewidth]{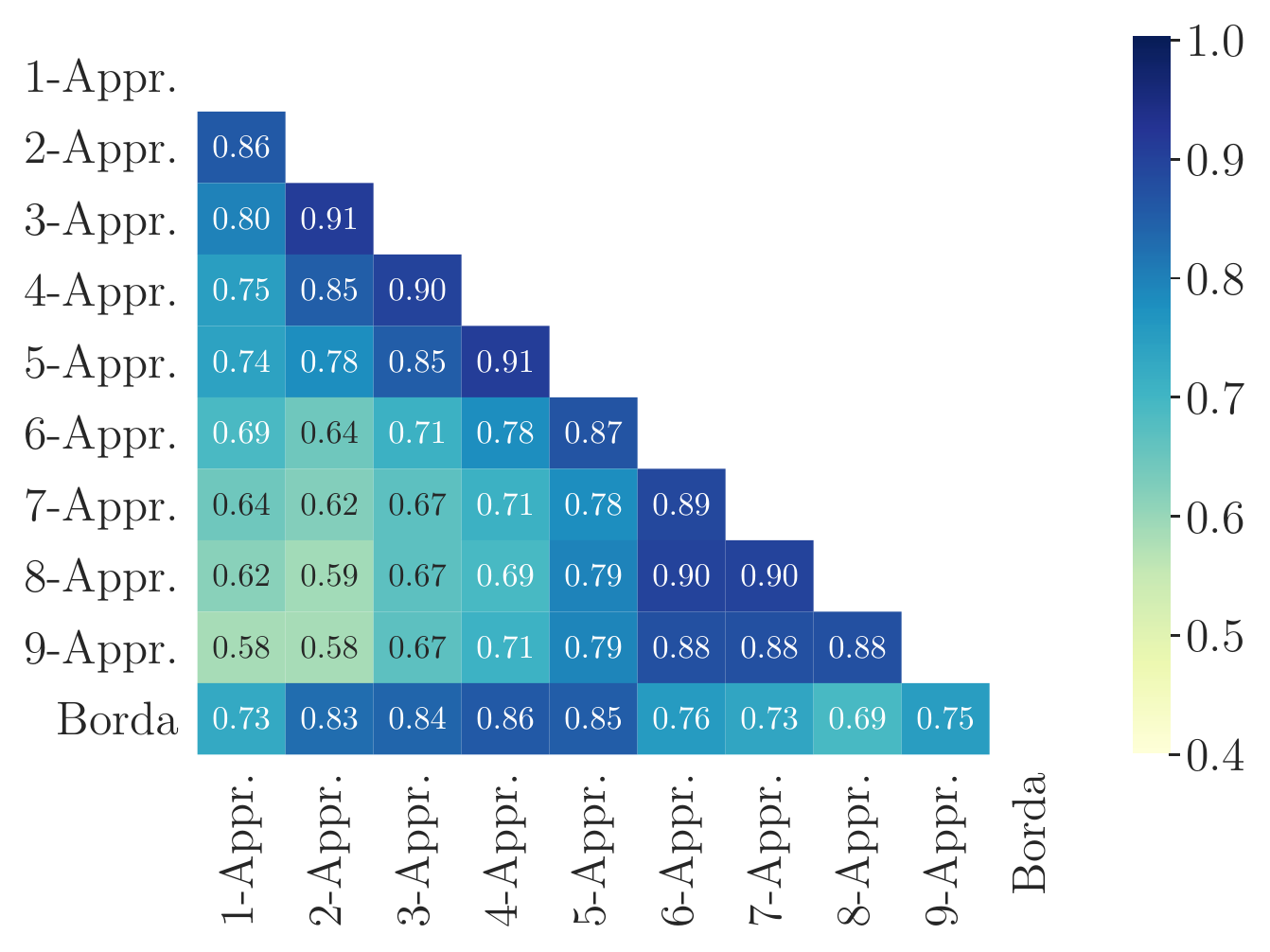}\vfill
		\caption{Task of selecting $W=3$ winners. }
	\end{subfigure}\hfill

	\begin{subfigure}[b]{.48\textwidth}
	\includegraphics[width=\linewidth]{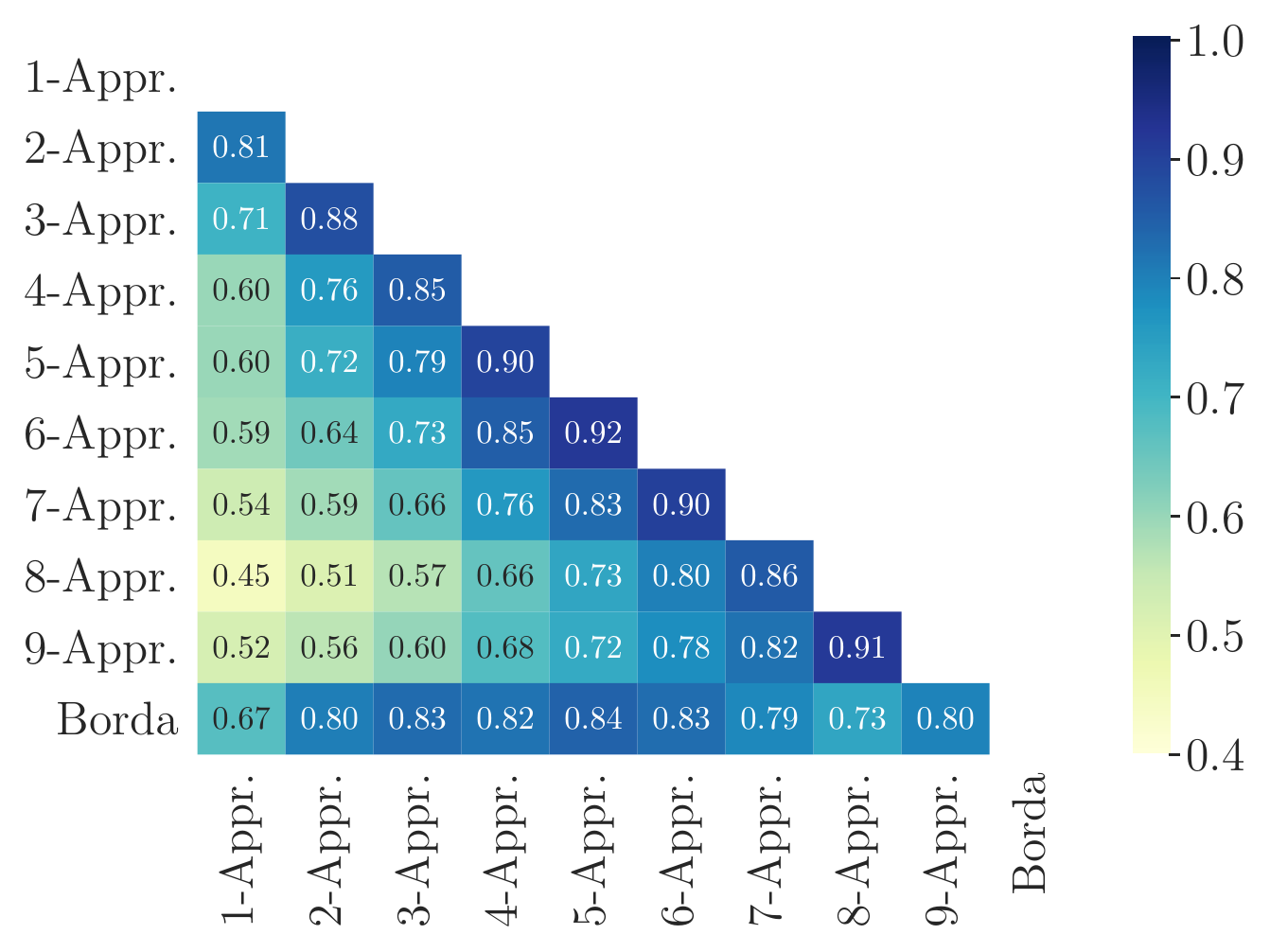}\vfill
	\caption{Task of ranking all candidates. The values plotted are the average Kendall's $\tau$ rank correlation between resulting rankings.}
\end{subfigure}\hfill

	\caption{ More approximate design invariance plots}
			\label{fig:approxdesigninvar_app}
\end{figure*}

\makeatletter{}%
\begin{table}
\begin{center}
\begin{tabular}{lclc}
\toprule
\textbf{Dep. Variable:}                   &  Best Mechanism  & \textbf{  R-squared:         } &     0.273   \\
\textbf{Model:}                           &       OLS        & \textbf{  Adj. R-squared:    } &     0.264   \\
\textbf{Method:}                          &  Least Squares   & \textbf{  F-statistic:       } &     42.68   \\
\textbf{Date:}                            & Wed, 12 Jun 2019 & \textbf{  Prob (F-statistic):} &  2.72e-11   \\
\textbf{Time:}                            &     16:17:32     & \textbf{  Log-Likelihood:    } &   -531.30   \\
\textbf{No. Observations:}                &         241      & \textbf{  AIC:               } &     1071.   \\
\textbf{Df Residuals:}                    &         237      & \textbf{  BIC:               } &     1085.   \\
\textbf{Df Model:}                        &           3      & \textbf{                     } &             \\
\bottomrule
\end{tabular}
\begin{tabular}{lcccccc}
                                          & \textbf{coef} & \textbf{std err} & \textbf{z} & \textbf{P$>$$|$z$|$} & \textbf{[0.025} & \textbf{0.975]}  \\
\midrule
\textbf{Intercept}                        &      -0.1687  &        0.411     &    -0.411  &         0.681        &       -0.973    &        0.636     \\
\textbf{Number Winners}                   &       0.9133  &        0.126     &     7.229  &         0.000        &        0.666    &        1.161     \\
\textbf{Number Candidates}                &       0.2662  &        0.057     &     4.630  &         0.000        &        0.154    &        0.379     \\
\textbf{Number Winners:Number Candidates} &      -0.0446  &        0.008     &    -5.786  &         0.000        &       -0.060    &       -0.030     \\
\bottomrule
\end{tabular}
\begin{tabular}{lclc}
\textbf{Omnibus:}       &  8.524 & \textbf{  Durbin-Watson:     } &    1.693  \\
\textbf{Prob(Omnibus):} &  0.014 & \textbf{  Jarque-Bera (JB):  } &    8.414  \\
\textbf{Skew:}          &  0.417 & \textbf{  Prob(JB):          } &   0.0149  \\
\textbf{Kurtosis:}      &  2.624 & \textbf{  Cond. No.          } &     463.  \\
\bottomrule
\end{tabular}
\caption{OLS Regression on the best K to use in K-Approval, by the number of candidates and desired winners. Standard errors are cluster standard errors, where each cluster is an election in our dataset.}
\label{tab:kregression}
\end{center}
\end{table} %

\begin{figure*}
	\centering
	\begin{subfigure}[b]{.56\textwidth}
		\includegraphics[width=\linewidth]{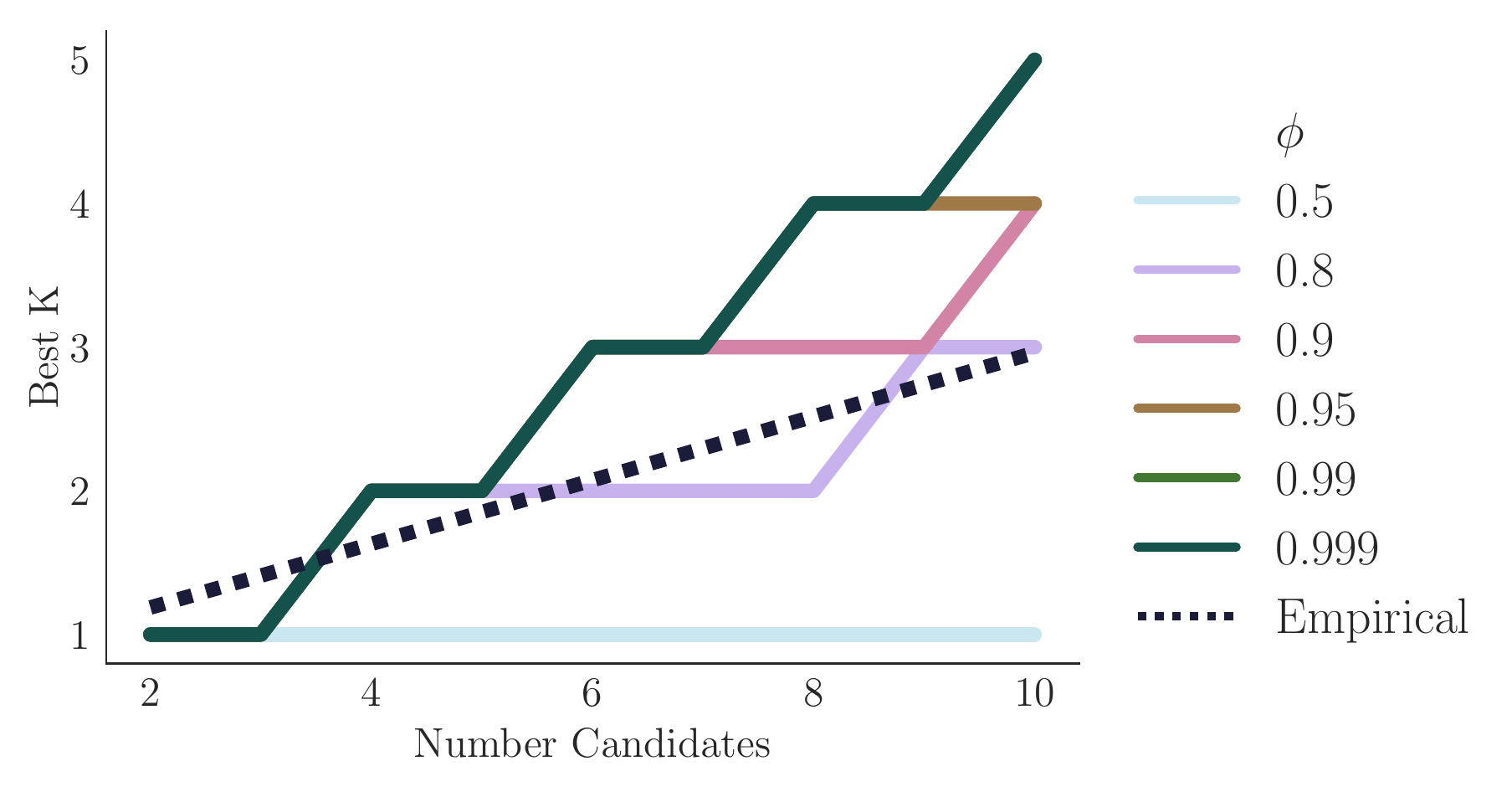}\vfill
		\caption{For selecting $W=1$ winner as number of candidates vary. }
		\label{fig:mallowsnumcandidapp}
	\end{subfigure}\hfill
	\begin{subfigure}[b]{.44\textwidth}
		\includegraphics[width=\linewidth]{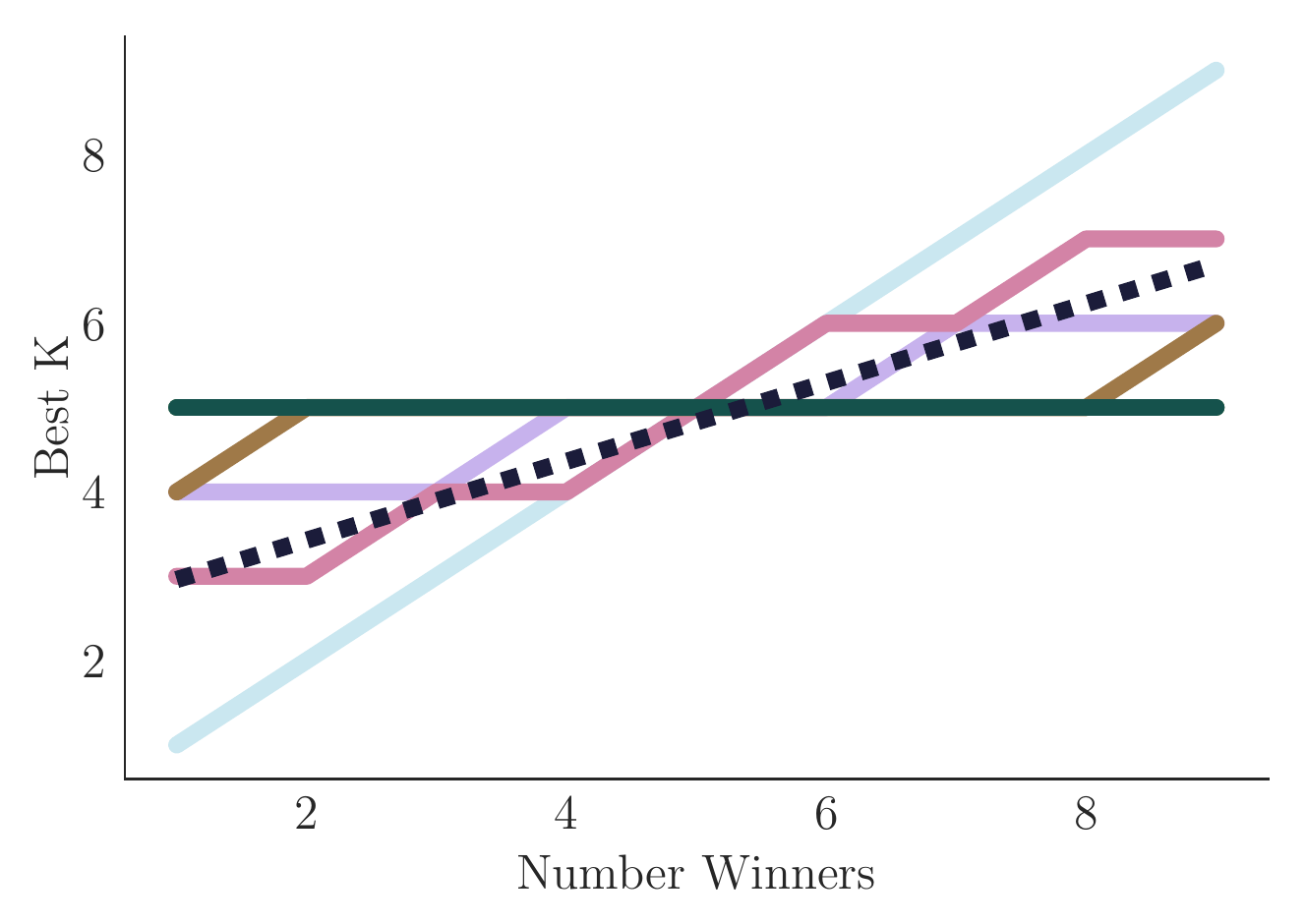}\vfill
		\caption{For $M=10$ candidates as number of winners vary.}
		\label{fig:mallowsnumwinnersapp}
	\end{subfigure}\hfill
	\caption{$K$-Approval rate optimal mechanism for the Mallows model as $\phi$, number of candidates, and number of winners vary. This plot contains an empirical line, which is calculated using the coefficients in the regression contained in Table~\ref{tab:kregression}. }
	\label{fig:mallowsplotsapp}
\end{figure*}

\begin{table}[h]
	\begin{center}
		\begin{tabular}{lcccc}
			\textbf{Election} & \textbf{Number Winners} & \textbf{Mechanism 1} & \textbf{Mechanism 2} & \textbf{Beats Approval rate optimal} \\
			Durham Ward 1, 2019 & 2 & 3 & 4 & \textbf{True} \\ %
			Durham Ward 1, 2019 & 13 & 8 & 9 & \textbf{True} \\ %
			Durham Ward 1, 2019 & 17 & 6 & 7 & \textbf{True} \\ %
			Irish03 & 1 & 1 & 3 & False \\
			Irish03 & 10 & 8 & 9 & False \\
			Irish03 & 10 & 8 & 10 & False \\
			Irish01 & 5 & 3 & 4 & \textbf{True} \\
			Irish01 & 5 & 3 & 5 & False \\
			Irish01 & 5 & 3 & 11 & False \\
			Irish01 & 7 & 1 & 5 & False \\
			Glasgow05 & 2 & 4 & 5 & \textbf{True} \\
			Glasgow05 & 2 & 5 & 6 & \textbf{True} \\
			Glasgow10 & 2 & 3 & 4 & False \\
			Glasgow10 & 2 & 4 & 6 & False \\
			Glasgow10 & 2 & 5 & 6 & \textbf{True} \\
			APA08 & 2 & 2 & 3 & \textbf{True} \\
		\end{tabular}
		\caption{Elections and goals where randomizing between two $K$-Approval mechanisms produces leads to faster learning than using either of the mechanisms separately. For several of these cases, randomization also beats the Approval rate optimal mechanism. }
		\label{tab:randomization}
	\end{center}
\end{table}

\FloatBarrier
%


%
\makeatletter{}%
\section{Proofs}

\subsection{Asymptotic design-invariance}

\thmconsistency*
\begin{proof}
	\begin{align*}
	\forall i \in C: \bbE[s_{iv}]
	&= \sum_{m = 1}^M \beta(m) \pr_F(\sigma_v(i) = m)\\
	&= \sum_{m = 1}^M \beta(m) \pr_F(\sigma_v(i) \leq m) - \sum_{m = 2}^M \beta(m) \pr_F(\sigma_v(i) < m) \\
	&= \sum_{m = 1}^M \beta(m) \pr_F(\sigma_v(i) \leq m) - \sum_{m = 1}^{M-1} \beta(m+1) \pr_F(\sigma_v(i) \leq m) \\
	&= \beta(M) + \sum_{m = 1}^{M-1} \left[ \beta(m)- \beta(m+1)\right] \pr_F(\sigma_v(i) \leq m)  
	\end{align*}
	~\\\\	$\implies$. By the definition of asymptotically design-invariant, 	$$ \exists O^*: \forall \beta \in \mB, \lim_{N \to \infty} \mO(M, N, F, \beta, G) = O^*, \text{ with probability } 1 $$
	For this $O^* = \{C_1^*, \dots, C_T^*\}$, we show by contradiction that $\forall s<t$: $i \in C^*_s, j\in C^*_t \implies$ $\pr_F(\sigma_v(i) \leq k) > \pr_F(\sigma_v(j) \leq k)$, $\forall k\in\{1 \dots M-1 \}$: Suppose $\exists i \in C^*_s, j\in C^*_t, s<t, k \in \{1 \dots M-1 \}$ such that $\pr_F(\sigma_v(i)\leq k) \leq \pr_F(\sigma_v(j)\leq k)$. Then, let 
	$$\beta(m) = \begin{cases} 1 & m \leq k \\
	0 & m > k	
	\end{cases}$$
	Then, $\bbE[s_{iv}] = \beta(M) + \beta(k)\pr_F(\sigma_v(i)\leq k) \leq \bbE[s_{jv}]$. Then, with positive probability, $$\lim_{N \to \infty} \mO(M, N, F, \beta, G) \neq O^*$$. 

	~\\$\impliedby$. Suppose there exists such a  $O^*$. Then, $\forall \beta \in \mB = \{\beta : \forall k<\ell \in {1 \dots M}, \beta(k)\geq \beta(\ell), \text{ and } \exists k<\ell, \beta(k)> \beta(\ell) \}$:
	Suppose $i\in C_s^*, j \in C_t^*, s < t$:
	\begin{align*}
	\bbE[s_{iv}] &= \beta(M) + \sum_{m = 1}^{M-1} \left[ \beta(m)- \beta(m+1)\right] \pr_F(\sigma_v(i) \leq m)\\
	&> \bbE[s_{jv}]
	\end{align*}
	Where the strict inequality follows as $\exists m: \beta(m) - \beta(m+1)>0$. Then, for all candidates $i\in C_s^*, j \in C_t^*, s < t$, by the strong law of large numbers $\lim_{N\to\infty} s_i^N > \lim_{N\to\infty} s_j^N$ w.p. $1$. Thus, $\lim_{N \to \infty} \mO(M, N, F, \beta, G) = O^*$ w.p. $1$. 
\end{proof}

\begin{remark}
	The following example, with candidates $A,B,C,D$ leads to a disjoint set of $2$ winners with $1$-Approval and $2$ approval, respectively
	\FloatBarrier
	\begin{table}[h]
		\begin{tabular}{l|ccccc}
			& Voter 1 & Voter 2 & Voter 3 & Voter 4 & Voter 5 \\ \hline
			Rank 1 & A       & A       & D       & D       & B       \\
			Rank 2 & B       & B       & C       & C       & C       \\
			Rank 3 & C       & C       & B       & B       & D       \\
			Rank 4 & D       & D       & A       & A       & A      
		\end{tabular}
	\end{table}
	\FloatBarrier
	With $1$-Approval, candidates $A,D$ are selected. With $2$-Approval, $B,C$ are selected.
	
\end{remark}

\subsection{Learning rates}
\parbold{Notation}

$t_{ij}^i(K)$ is the probability that $i$ is approved but $j$ is not, using $K$-Approval.

For convenience, we overload the rate function $r(\cdot)$:
\begin{itemize}
	\item $r_{ij}(\beta)$ is as defined in Proposition~\ref{thm:pairwiselearning}, the large deviation rate to learn a pair of candidates $i,j$ given scoring rule $\beta$, for a fixed $F$ that should be clear from context. When a goal $G$ is clear from context, $r(\beta)$ is as defined in Proposition~\ref{thm:goal_learning}, the minimum over $r_{ij}(\beta)$ for candidate pairs that are in different asymptotic tiers. 
	\item $r_{ij}(K)$ is as defined in Proposition~\ref{lem:pairwiselearning_approval}, the large deviation rate to learn a pair of candidates $i,j$ using $K$-Approval. $r(K)$ is analogous to the previous item when using $K$ approval. 
	\item $r(a,b)$ is the large deviation rate to learn a pair of candidates $i,j$ using approval voting when the probability that $i$ is approved but $j$ is not is $a$; and $b$ is the probability that $j$ is approved but $i$ is not is not. 	
\end{itemize}
When which rate function we mean is clear from context, we may drop the argument $(\cdot)$ and just write $r_{ij}$ or $r$.

\begin{remark}
	$r(a,b) > r(c,d)$ when $a> c, b \leq d$, OR $a \geq c, b < d$. 
	\label{rem:tminusaplusa}
\end{remark}
\begin{proof}
	$\gamma(a,b) = \sqrt{ab} + 1 - a - b$ is strictly concave in $a,b$, with maximum at $a=b$. Thus, holding either $a$ or $b$ constant and moving the other farther away strictly decreases $\gamma$, and thus strictly increases $r$. 
\end{proof}

\thmpairwiselearning*
\begin{proof}
	
	Define the following random variable for each voter $v \sim F$: 
	\begin{align*}
	A_v = \beta(\sigma_v(i)) - \beta(\sigma_v(j))
	\end{align*}
	Then, $\sigma^N(i) < \sigma^N(j)$ when $A^N = \sum_{v=1}^N A_v > 0$, and $\bbE[A_v]>0$ by supposition. 
	Let
	\begin{align*}
		r_{ij}(\beta) &= -\inf_{z\in\bbR}{\Lambda(z)}\\
	\Lambda(z) &= \log \left[\sum_{m=1}^M \sum_{\ell\neq m} \pr(\sigma_v(i) = m, \sigma_v(j)=\ell) \exp[ z(\beta(\sigma_v(i)) - \beta(\sigma_v(j)))]   \right]
	\end{align*}
	
	Then, by basic large deviation bounds (see, e.g.~\citet{dembo_large_2010}):
	\begin{align*}
	-\lim_{N\to \infty} \frac{1}{N} \log \pr(A^N \leq 0) &= r_{ij}(\beta)\\
	\end{align*}

	And, applying Chernoff bounds, we get the standard relationship to the large deviation rate, giving an upper bound for the probability of error directly, including any polynomial factors out front:
		\begin{align*}
	\pr(\sigma^N(i) > \sigma^N(j)) &\leq \pr(A^N \leq 0) \\
	&< \left[\inf_{z>0}\bbE[\exp[-zA_v]]\right]^N\\
	&= \left[\inf_{z<0}\left[\sum_{m=1}^M \sum_{\ell\neq m} \pr(\sigma_v(i) = m, \sigma_v(j)=\ell) \exp[ z(\beta(\sigma_v(i)) - \beta(\sigma_v(j)))]\right]\right]^N\\
		&= \left[\inf_{z<0} \exp[\Lambda(z)]\right]^N\\
		&= \exp[-r_{ij}(\beta)N]
	\end{align*}
	
Then, $\pr(\sigma^N(i) < \sigma^N(j)) > 1-\epsilon$ when
\begin{align*}
\exp[-r_{ij}N] &< \epsilon\\
\iff N&> \frac{1}{r_{ij}}\log\left(\frac{1}{\epsilon}\right)
\end{align*}
	
\end{proof}

\lempairwiseapproval*
\begin{proof}
	With $K$-approval voting, $A^{ij}_v$ becomes
	\begin{align*}
	A_v = \begin{cases}
	1 & \text{w.p. } t^i_{ij}, \text{ i.e., when candidate $i$ approved but $j$ not approved}\\
	0 & \text{w.p. } 1-t^i_{ij} - t^j_{ij}, \text{ i.e., when both approved, or neither approved}\\
	-1 & \text{w.p. } t^j_{ij}, \text{ i.e., when candidate $j$ approved but $i$ not approved}
	\end{cases}
	\end{align*}
Then,
	\begin{align*}
r_{ij} &= -\inf_{z\in\bbR}{\Lambda(z)}\\
\Lambda(z) &= \log \left[\sum_{m=1}^M \sum_{\ell\neq m} \pr(\sigma_v(i) = m, \sigma_v(j)=\ell) \exp[ z(\beta(\sigma_v(i)) - \beta(\sigma_v(j)))]   \right]\\
 &= \log \left[t_{ij}^i \exp(z) + t_{ij}^j \exp(-z) + (1 - t_{ij}^i - t_{ij}^j)\right]
\end{align*}
The $\inf(\Lambda)$ is attained at $z = \frac{1}{2} \log\frac{t_{ij}^j}{t_{ij}^i}$ ($\Lambda$ is convex in $z$, and so setting the first derivative to zero finds the $\inf$). 
And so
\begin{align*}
r_{ij} &= -\log\left[ t_{ij}^i \exp\left(\frac{1}{2} \log\frac{t_{ij}^j}{t_{ij}^i}\right) + t_{ij}^j \exp\left(-\frac{1}{2} \log\frac{t_{ij}^j}{t_{ij}^i}\right) + (1 - t_{ij}^i - t_{ij}^j) \right]\\
&= -\log\left[ 2\sqrt{t_{ij}^it_{ij}^j} + 1 - t_{ij}^i - t_{ij}^j \right]
	\end{align*}
\end{proof}

\thmgoallearning*
\begin{proof}

By the Union bound
\begin{align*}
Q^N &= \sum_{i\in C_s^*,j\in C_t^*, s < t} \pr(\sigma^N(i) > \sigma^N(j)) \\
&\leq \sum_{i\in C_s^*,j\in C_t^*, s < t} \exp[-r_{ij}N] & \text{Proposition }~\ref{thm:pairwiselearning}\\
&\leq M^2  \exp[-rN]\\
\end{align*}

Now, using large deviation properties:

By supposition, $Q^N \to 0$, and so $-Q^N$ approaches $0$ from below.
Then,
\begin{align}
-\lim_{N\to\infty} \frac{1}{N}\log(Q^N)
&=  -\lim_{N\to\infty} \frac{1}{N}\log \sum_{i\in C_s^*,j\in C_t^*, s < t} \pr(\sigma^N(i) > \sigma^N(j))\nonumber\\
&= -\max_{i\in C_s^*,j\in C_t^*, s < t}\left(\lim_{N\to\infty} \frac{1}{N}\log \pr(\sigma^N(i) > \sigma^N(j))\right)\label{eqnstep:ldpproperty}\\
&= \min_{i\in C_s^*,j\in C_t^*, s < t}-\left(\lim_{N\to\infty} \frac{1}{N}\log \pr(\sigma^N(i) > \sigma^N(j))\right)\nonumber\\
&= \min_{i\in C_s^*,j\in C_t^*, s < t} r_{ij} = r\nonumber
\end{align}

Line~\eqref{eqnstep:ldpproperty} follows from: $\forall a^\epsilon_i \geq 0$, ${\lim \sup}_{\epsilon \to 0} \left[\epsilon\log\left(\sum_i^N a^\epsilon_i\right)\right] = \max^N_i {\lim \sup}_{\epsilon \to 0} \epsilon \log(a^\epsilon_i)$.
See, e.g., Lemma 1.2.15 in~\cite{dembo_large_2010} for a proof of this property. 

Thus $r$ is the large deviation rate for $Q^N$. 

\end{proof}

\subsection{Design insights}

\lemrandomizebetterscoring*
\begin{proof}
	From Proposition~\ref{thm:pairwiselearning}, for a given scoring rule $\beta$ and pair of candidates $i,j$, the learning rate is
	\begin{align*}
	r_{ij}(\beta) &= -\inf_{z\in\bbR} \log \bbE_F\left[\exp\left[z\left[\beta(\sigma_v(i)) - \beta(\sigma_v(j))\right]\right]\right]\\
	&= - \log \inf_{z\in\bbR} \bbE_F\left[\exp\left[z\left[\beta(\sigma_v(i)) - \beta(\sigma_v(j))\right]\right]\right]
	\end{align*}
	
	Similarly, if we use scoring rules $\{\beta^u\}_{u=1}^P$, each with probability $d^u$, then, 
	\begin{align*}
	r_{ij}(\{\beta^u\}_{u=1}^P) &= - \log \inf_{z\in\bbR} \bbE_{F,\{d^u,\beta^u\}}\left[\exp\left[z\left[\beta^u(\sigma_v(i)) - \beta^u(\sigma_v(j))\right]\right]\right]\\
	&= - \log \inf_{z\in\bbR} \sum_u d^u  \bbE_{F}\left[\exp\left[z\left[\beta^u(\sigma_v(i)) - \beta^u(\sigma_v(j))\right]\right]\right]
	\end{align*}

	Now, for a single scoring rule $\beta(\cdot)$, let \[\gamma(\beta(1) , \dots, \beta(M)) \triangleq  \bbE_F\left[\exp\left[z\left[\beta(\sigma_v(i)) - \beta(\sigma_v(j))\right]\right]\right]\]

	Below, we show that $\gamma(\beta(1) , \dots, \beta(M))$  is convex in $\beta(k), \forall k, z$. Then, by convexity, $\forall z$
	\begin{align*}
	\sum_{u=1}^{P} d^u \gamma(\beta^u(1), \dots, \beta^u(M))& \geq
	\gamma\left(\sum_{u=1}^{P} d^u\beta^u(1) , \dots, \sum_{u=1}^{P} d^u\beta^u(M)\right)
	\end{align*}
	and so
	\begin{align*}
	\inf_z\left[\sum_{u=1}^{P} d^u \gamma(\beta^u(1), \dots, \beta^u(M))\right]& \geq
	\inf_z \gamma\left(\sum_{u=1}^{P} d^u\beta^u(1) , \dots, \sum_{u=1}^{P} d^u\beta^u(M)\right)
	\end{align*}
	
	The left hand side is equal to the argument inside the $-\log(\cdot)$ for the rate function for randomizing between scoring rules $\{\beta^u\}_{u=1}^P$, each with probability $d^u$, and the right hand side is the argument inside for the rate function for instead using the single scoring rule $\beta^*$ defined as the convex combination of $\{\beta^u\}_{u=1}^P$. Then, as $-\log(x)$ is decreasing in $x$, we have that \[r_{ij}(\beta^*) \geq r_{ij}(\{\beta^u\}_{u=1}^P)\].
	
	As this holds for each pair of candidates $i,j$ simultaneously, we are done.

	\parbold{Proof that  $\gamma(\beta(1) , \dots, \beta(M)) \triangleq  \bbE_F\left[\exp\left[z\left[\beta(\sigma_v(i)) - \beta(\sigma_v(j))\right]\right]\right]$ is convex in $\beta(k), \forall k$}
	
	We directly calculate the Hessian of $\gamma$ and note that it is diagonally dominant and thus positive semidefinite. For notational convenience, we let $\beta_k = \beta(k)$, and $\sigma(k,\ell) = \pr_F(\sigma_v(i) = k, \sigma_v(j) = \ell)$. Of course, $\sigma(k,k)=0$, as we assume each voter has a strict ranking as her preference.
	\begin{align*}
	\gamma(\beta_1,\dots,\beta_M) &=   \bbE_F\left[\exp\left[z\left[\beta(\sigma_v(i)) - \beta(\sigma_v(j))\right]\right]\right]\\
	&= \sum_{k=1}^{M}\sum_{\ell = 1}^M \sigma(k,\ell) \exp\left[z \left[\beta_k - \beta_\ell \right]\right]\\
	\frac{\partial}{\partial \beta_k} \gamma(\beta_1,\dots,\beta_M)	&= z\exp\left[z \beta_k\right]\sum_{\ell\neq k}\exp\left[-z \beta_\ell\right]\sigma(k,\ell) - z\exp\left[-z \beta_k\right]\sum_{\ell\neq k}\exp\left[z \beta_\ell\right]\sigma(\ell,k)\\
	\frac{\partial^2}{\partial \left(\beta_k\right)^2} \gamma(\beta_1,\dots,\beta_M)	&= z^2\exp\left[z \beta_k\right]\sum_{\ell\neq k}\exp\left[-z \beta_\ell\right]\sigma(k,\ell) + z^2\exp\left[-z \beta_k\right]\sum_{\ell\neq k}\exp\left[z \beta_\ell\right]\sigma(\ell,k)\\
	&= z^2 \left[\left[\exp\left[z \beta_k\right]\sum_{\ell\neq k}\exp\left[-z \beta_\ell\right]\sigma(k,\ell)\right] + \left[\exp\left[-z \beta_k\right]\sum_{\ell\neq k}\exp\left[z \beta_\ell\right]\sigma(\ell,k)\right]\right]\\
	\frac{\partial^2}{\partial \beta_k\beta_\ell}\gamma(\beta_1,\dots,\beta_M)&= -z^2\left[\exp\left[z \beta_k\right]\exp\left[-z \beta_\ell\right]\sigma(k,\ell) + \exp\left[-z \beta_k\right]\exp\left[z \beta_\ell\right]\sigma(\ell,k)\right]
	\end{align*}
	Thus, the Hessian of $\gamma$ is diagonally dominant with non-negative diagonal elements: $\forall k$,
	$$ \left|\frac{\partial^2 \gamma}{\partial \left(\beta_k\right)^2}  \right| \geq \sum_{\ell\neq k}\left|\frac{\partial^2 \gamma}{\partial \beta_k\beta_\ell}  \right|$$
	and so the Hessian is positive semi-definite. Thus, $\gamma$ is convex in $\beta_k$. 

\end{proof}
\lemrandomizenotbetterapproval*
\begin{proof}
	From Proposition~\ref{lem:pairwiselearning_approval}, for $k$-Approval, 
	\begin{align*}
		r_{ij}(t^i_{ij}, t^j_{ij}) &= -\log\left[ 2\sqrt{t_{ij}^it_{ij}^j} + 1 - t_{ij}^i - t_{ij}^j \right]
	\end{align*}
	Where $t^i_{ij}$ is the probability that $i$ is approved but $j$ is not.
	
	This rate function is convex in $t^i_{ij}, t^j_{ij}$:
	\begin{itemize}
		\item $r_{ij}(a,b) = h(g(a, b))$, where $h(x) = -\log(x)$, $g(a, b) = 2\sqrt{ab} + 1 - a - b$. 
		\item $g(a, b)$ is concave in $a, b$
		\item $h$ is convex, and $\tilde{h}$ is non-increasing, where $\tilde{h}(x) =  \begin{cases} h(x) & x > 0 \\ \infty & x \leq 0 \end{cases}$ is the extended value function of $h$. 
		\item By convex composition rules, $r_{ij}(a,b)$ is convex (see, e.g., page 84 of~\citet{boyd_convex_2004}). 
	\end{itemize}
	
	The result follows by convexity. Consider a randomization of $K$-Approval mechanisms for $K\in\{1, \dots, M-1\}$, where $K$-Approval is used with probability $d^K$.

	The resulting approval probabilities are: $t^i_{ij} = \sum_{K=1}^{M-1} d^K t^i_{ij}(K), t^j_{ij} = \sum_{K=1}^{M-1} d^K t^j_{ij}(K)$. By convexity:
	\begin{align*}
	r\left(\sum_{K=1}^{M-1} d^K t^i_{ij}(K), \sum_{K=1}^{M-1} d^K t^j_{ij}(K)\right) &\leq \sum_{K=1}^{M-1} d^K r\left(t^i_{ij}(K), t^j_{ij}(K)\right)\\
	&= \sum_{K=1}^{M-1} d^K  r_{ij}(K)\\
	&\leq \max_K r_{ij}(K)
	\end{align*}

We note that, unlike the previous proof, we cannot conclude in general that randomization cannot improve the rates at which the outcome is learned (in fact, Theorem~\ref{lem:randomizebetterapproval_Wselection} establishes otherwise). That is because while the same $\beta^*$ could be said to be rate optimal (compared to the randomized mechanism) for every pair of candidates simultaneously in that proof, in this proof $\arg\max_K r_{ij}(K)$ may change based on the pair $i,j$.

\end{proof}

\lemmallowsnorando*
\begin{proof}
	When selecting $W$ winners out of $M$ candidates, we need to separate candidates $1 \dots W$ from candidates $W + 1 \dots M$. It is easy to show that the pivotal pair, regardless of which $K$ is used in $K$-Approval, is $W$,$W+1$. Applying Theorem~\ref{lem:randomizenotbetterapproval}, then, randomization cannot help the overall rate.
\end{proof}

\lemmallowsnotWK*
\begin{proof}
	
	We prove the result by providing an example where it is not optimal. Suppose there are 4 candidates, and we wish to select $3$ winners, i.e., separate the first three items from the last item.

Let the items in the reference ranking be, in order, $i=1, 2, 3, 4$, respectively.

A Mallows model (with parameter $\phi = \frac{p}{1-p}$, where $p$ is the probability of flipping a given pair of candidates) can be sampled by repeated insertion \citep{lu14a,diaconis_group_nodate}: starting from the first item in the reference ranking, there exists probability $\tilde{p}_{ij} = \frac{\phi^{i-j}}{1 + \phi^1 + \dots + \phi^{i-1}}$ at which item $i$ can be inserted into position $j\leq i$, \textit{independently} of how items above it were inserted, such that the resulting ranking distribution matches the Mallows model.

Using this repeated insertion property for our example, we can derive $p_{\ell k}$, the probability at which item $3$ is in position $\ell$ and item $4$ is in position $k$ after sampling from a Mallows model with parameter $\phi$. 

In particular, if $\ell < k$, $p_{\ell k}$ is exactly the probability that item $3$ is inserted in position $\ell$ and item $4$ is inserted in position $k$. If $k>\ell$, however, it is the probability that item $3$ is inserted in position $\ell - 1$ and then pushed down when item $4$ is inserted in position $k$. (More generally, it turns out, the exact probability for an item appearing in a given position in the Mallows model can be calculated using a simple dynamic program, a fact that does not appear to be documented  elsewhere but may be independently useful. We used this dynamic program to find this given example). 

Then, for our example
\begin{align*}
N_i &\triangleq 1 + \phi^1 + \dots + \phi^{i-1} \\
p_{\ell k} &= \begin{bmatrix}
0 & \frac{\phi^2}{N_3}\frac{\phi^2}{N_4} & \frac{\phi^2}{N_3}\frac{\phi^1}{N_4} & \frac{\phi^2}{N_3}\frac{\phi^0}{N_4}\\
\frac{\phi^2}{N_3}\frac{\phi^3}{N_4} & 0 & \frac{\phi^1}{N_3}\frac{\phi^1}{N_4} & \frac{\phi^1}{N_3}\frac{\phi^0}{N_4}\\
\frac{\phi^1}{N_3}\frac{\phi^3}{N_4} & \frac{\phi^1}{N_3}\frac{\phi^2}{N_4} & 0 & \frac{\phi^0}{N_3}\frac{\phi^0}{N_4}\\
\frac{\phi^0}{N_3}\frac{\phi^3}{N_4} & \frac{\phi^0}{N_3}\frac{\phi^2}{N_4} & \frac{\phi^0}{N_3}\frac{\phi^1}{N_4} & 0\\
\end{bmatrix}_{\ell k}
=
\frac{1}{N_3N_4}\begin{bmatrix}
0 & \phi^4 & \phi^3 & {\phi^2}\\
\phi^5 & 0 & \phi^2 & {\phi^1}\\
\phi^4 & \phi^3 & 0 & 1\\
{\phi^3} & {\phi^2} & {\phi^1} & 0\\
\end{bmatrix}_{\ell k}
\end{align*}

Recall that $t_{ij}^i(K)$ is the probability that $i$ is approved but $j$ is not, using $K$-Approval. Then, if we use $3$-approval and $2$-approval, respectively:
\begin{align*}
N_3 &= 1 + \phi + \phi^2\\
N_4 &= 1 + \phi + \phi^2+ \phi^3\\
t_{34}^3(3) &= p_{14} + p_{24} + p_{34}= \frac{\phi^2  + \phi^1 + 1}{N_3N_4}= \frac{1}{N_4}\\
t_{34}^4(3) &= p_{41} + p_{42} + p_{43} = \frac{\phi^3  + \phi^2 + \phi^1}{N_3N_4} = \frac{\phi}{N_4}\\
t_{34}^3(2) &= p_{14} + p_{24} + p_{13} + p_{23} = \frac{\phi^2  + \phi^1 + \phi^3 + \phi^2}{N_3N_4}\\
t_{34}^4(2) &= p_{41} + p_{42} + p_{31} + p_{32} = \frac{\phi^3  + \phi^2 + \phi^4 + \phi^3}{N_3N_4}\\
\end{align*}

Then, recall the rate between items $i,j$ using $K$ approval is
\begin{align*}
	r_{ij}(K) &= -\log\left[ 2\sqrt{t_{ij}^i(K)t_{ij}^j(K)} + 1 - t_{ij}^i(K) - t_{ij}^j(K) \right]\\
	 r_{34}(3) &= 	-\log\left[ 2\sqrt{\frac{1}{N_4}\frac{\phi}{N_4}} + 1 - \frac{1}{N_4} - \frac{\phi}{N_4} \right]\\
	 r_{34}(2) &= 	-\log\left[ 2\sqrt{\frac{\phi^2  + \phi^1 + \phi^3 + \phi^2}{N_3N_4}\frac{\phi^3  + \phi^2 + \phi^4 + \phi^3}{N_3N_4}} + 1 - \frac{\phi^2  + \phi^1 + \phi^3 + \phi^2}{N_3N_4} - \frac{\phi^3  + \phi^2 + \phi^4 + \phi^3}{N_3N_4} \right]
\end{align*}

When there is low noise, e.g., $\phi = .1$ ($p = .091$):
\begin{align*}
r_{34}(3) &= .5462\\
r_{34}(2) &= .04696 < r_{34}(3)
\end{align*}

But when there is high noise, e.g., $\phi = .8$ ($p=.44$):
\begin{align*}
r_{34}(3) &= .00378\\
r_{34}(2) &= .00402 > r_{34}(3)
\end{align*}

Note that the same example works for selecting 1 winner out of the 4 candidates, as the repeated insertion model can be run in reverse.

\end{proof}

\lemrandomizebetterapprovalWselection*
\begin{proof}
We provide two proofs: a numeric example from a real-world election, and a contrived, constructed example.

\parbold{Numeric example found in a real election} In Durham Ward 1, to select 2 winners, randomizing between 3 and 4-Approval is better than either individually, even though asymptotically the mechanisms pick the same set of winners. The critical pair with $2$-Approval is with the candidate asymptotically ranked 1st, and the best item not selected. With $3$-Approval, it is with the candidate asymptotically ranked 2nd, and the same best item not selected.

We will call these items $h$, $i,j$ (the one not selected) respectively. The respective probabilities of being selected alone:

\begin{align*}
t_{hj}^h(3)&= 0.277 \\
t_{hj}^h(4)&= 0.266	\\
t_{hj}^j(3)&= 0.200	\\
t_{hj}^j(4)&= 0.188	\\
t_{ij}^i(3)&= 0.255	\\
t_{ij}^i(4)&= 0.295	\\
t_{ij}^j(3)&= 0.160\\
t_{ij}^j(4)&= 0.217\\\\
t_{hj}^h(\{3, 4\})&= 0.271	 \\
t_{hj}^j(\{3, 4\})&= 0.194	\\
t_{ij}^i(\{3, 4\})&= 0.275	\\
t_{ij}^j(\{3, 4\})&= 0.189\\
\end{align*}

And the resulting rates (using the formula in Proposition~\ref{lem:pairwiselearning_approval}) are:

\begin{align*}
r_{hj}(3) &= .00616932\\
r_{ij}(3) &= .01114061\\
r_{hj}(4) &= .00677352\\
r_{ij}(4) &= .00592327\\
r_{hj}(\{3, 4\}) &= .00642839\\
r_{ij}(\{3, 4\}) &= .00815633\\
r(3) &= \min(r_{hj}(3), r_{ij}(3)) = .00616932\\
r(4) &= \min(r_{hj}(4), r_{ij}(4)) = .00592327\\
r(\{3, 4\}) &= \min(r_{hj}(\{3, 4\}), r_{ij}(\{3, 4\})) = .00642839 > \max(r(3), r(4))
\end{align*}

Thus randomization improves learning.

\parbold{Constructed example with design invariance}
	We now construct a fully design-invariant example with the same flavor as the numeric example, where which pair is critical changes with the mechanism.
	
	Consider three candidates $h,i,j$, such that $h$ is asymptotically in the set of $W$ winners and $i,j$ are not. Thus, we need the rates at which $h$ is separated from both $i,j$. Let $W=K < L=K+1$.
	
	We prove the result by giving an example where: it is easier to separate $h$ from $i$ using $K$-Approval, and easier to separate $h$ from $j$ using $L$-Approval. Using $K$-Approval, $r_{hj}$ asymptotically dominates the rate at which the overall outcome is learned, and using $L$-Approval, $r_{hi}$ does. Further, randomizing between the two mechanisms improves the two rates that dominate enough such that the overall rate is improved. 
	
	We need to show the following hold for our example: one of the rates between candidates $h$ and $i,j$ are smaller than other rates, i.e., dominate the overall learning rate when $K$ and/or $L$ approval is used; randomization between $K$ and $L$ approval helps the minimum rate between candidates $h$ and $i,j$; $K'$-Approval ($K'\neq K, K' \neq L$) produces a worse rate than either $K$ or $L$ approval; and this example is asymptotically design-invariant.
	
	We prove each of these conditions in turn after specifying the example.
	
	Recall that $t_{ij}^i(k)$ is the probability that $i$ is approved but $j$ is not, using $k$-Approval. Here, we will use:	 $$t_{hi}^i(K),t_{hi}^i(L),t_{hi}^h(K),t_{hi}^h(L),
	t_{hj}^j(K),t_{hj}^j(L),t_{hj}^h(K),t_{hj}^h(L)$$. The end row labeled ``Total value'' then sums up these values.

	\parbold{Example Specification}
	Consider $F$ such according to the following table, where the first column is the probabilities of the positions in the second set of columns. The third set of columns indicates whether those set of positions contribute to the given probabilities, for easy accounting. 
\FloatBarrier
	\begin{table}[h]
		\hspace*{-1cm}
		\begin{tabular}{c|c|ccc||cccccccc}
			& & \multicolumn{3}{c||}{Positions of $h,i,j$} & \multicolumn{8}{c}{Contributes to? (Y = Yes)}\\
			Row& $Pr_F(\cdot)$ & $\sigma(h)$          & $\sigma(i)$          & $\sigma(j)$          &$t_{hi}^h(K)$ & $t_{hi}^i(K)$ & $t_{hi}^h(L)$ & $t_{hi}^i(L)$ & $t_{hj}^h(K)$ & $t_{hj}^j(K)$ & $t_{hj}^h(L)$ & $t_{hj}^j(L)$\\ \hline
		1&	$a $   &  $K$    &   $L+1$    & $L$    &Y&&Y&&Y&&&\\ %
		2&	$a$   &  $K-1$    &   $K$    &  $L$   &&&&&Y&&&\\
		3&	$T_1 - {a} - \epsilon$   &  $K$    &   $L+1$    &    $L+2$     &Y&&Y&&Y&&Y&\\
		4&	$T_2-2a$ & $L+1$   &    $L+2$     &   $K$  &&&&&&Y&&Y\\
		5&	$T_2-2a$ & $L+1$   &    $K$     &   $L+2$  &&Y&&Y&&&&\\
		6&	$a$	&   $L+1$      &     {$K$}       &   $L$  &&Y&&Y&&&&{Y}\\
		7&	 $a$	&   $L$      &    $K-1$    &    {$K$}     &&Y&&&&Y&&\\
		8&	 $a$	&   $L$      &   $L+1$ & $L+2$       &&&Y&&&&{Y}&\\ %
		9&	 $a$	&   $L+1$      &    $L+2$    &  {L}       &&&&&&&&Y\\
		10&	$\epsilon$ & \multicolumn{3}{c||}{See caption} &Y&&Y&&Y&&{Y}&\\
		11&	 $0$ & \multicolumn{3}{c||}{Otherwise}&&&&&&&&\\ \hline
			 \multicolumn{5}{r}{\textbf{Total value:}}&$T_1$&$T_2$&$T_1+a$&$T_2-a$&$T_1+a$&$T_2-a$&$T_1$&$T_2$
		\end{tabular} %
\caption{Where the constants such that  $0<\epsilon<a<\frac{T_2}{2}<T_2 < T_1<T_1 + 2T_2 + a=1$, i.e., the table describes a valid probability distribution.%
	Row 10 is as follows: The first $K$ candidates (the asymptotic winners) occupy the first $K$ spots, in an order drawn uniformly at random. Similarly, The bottom $M-K$ candidates occupy the bottom $M-K$ spots, in an order drawn uniformly at random. This randomization ensures asymptotically design invariance.
}
	\end{table}
\FloatBarrier

The table does not specify the probabilities of other candidates appearing in any position, so it is possible that they dominate the learning rate (are hardest to learn). (In particular, if the same, asymptotically non-winning candidate $q$ is always in position $L$ in the case in row $3$, then it may be hard to separate it from candidate $h$ using $L$ approval). However, we can specify the example further to ensure this does not happen.

Suppose candidates are indexed by their order in some strict ranking $\sigma^*$. Then, candidates $h=K, i = L=K+1, j = K+2$. Further suppose that candidates in $\{1 \dots K-1 \}$ always occupy, in order except in case of row 10, the best positions in a voter's ranking that are not reserved for candidates $h,i,j$ in the table above. %

For candidates $q\in\{K+3 \dots M \}$, we have to be more careful to avoid the case in parenthesis above. Suppose these $Q = M-K+2$ candidates fill up the bottom spots in a voter's ranking in a uniform at random order. In other words, they occupy spots $L+3 \dots M$, and the worst spot among whichever of $K, L, L+1,L+2$ is missing in each row in the table above.

\parbold{Rates between the $h$ and $i,j$ dominate the overall learning rate using $K$ or $L$ approval} 

We are now ready to show the first claim that learning between candidates $h$ and $i,j$ is hardest (when using either $K$ or $L$ approval).

By the specification above, candidates in $\{1 \dots K-2 \}$ are always approved, and so learning between those candidates and any non-winning candidate is faster than any large deviations rate. Similarly, candidate $K-1$ always is ranked higher than candidates $q\in\{K+3 \dots M \}$, and it is approved alone with high enough probability. %

Then, the other candidates who may dominate the learning rate are candidate $K-1$ (in separation from $i, j$), or $q\in\{K+3 \dots M \}$ (in separation from $h$). From the above table:
	\begin{align*}
	t_{hq}^{h}(K) &= T_1+a	& \text{Rows 1,2,3,10}\\
	t_{hq}^{q}(K) &= \frac{2a}{Q}	& \text{Rows 8,9} \\
	t_{hq}^{h}(L) &= \frac{Q-1}{Q}\left[T_1 - \epsilon \right] + 3a + \epsilon	& \text{Rows 1,2,7,10; and 3,8 w.p. } \frac{Q-1}{Q} \\ 
	t_{hq}^{q}(L) &= \frac{2T_2 - 3a}{Q}	& \text{Rows 4,5,9 w.p.}  \frac{1}{Q}\\
t_{(K-1)i}^{K-1}(K) &= T_1+T_2 & \text{Rows } 1,3,4,8,9,10\\
t_{(K-1)i}^{i}(K) &= 2a	& \text{Rows } 2,7\\
t_{(K-1)j}^{K-1}(K) &= T_1 + T_2+a	& \text{Rows } 1,3,5,6,8,9,10\\
t_{(K-1)j}^{j}(K) &= a	& \text{Rows } 7\\
t_{(K-1)i}^{K-1}(L) &= T_1+T_2 & \text{Rows } 1,3,4,8,9,10\\
t_{(K-1)i}^{i}(L) &= 2a	& \text{Rows } 2,7\\
t_{(K-1)j}^{K-1}(L) &= T_1 + T_2 - 2a	& \text{Rows } 3,5,8,10\\
t_{(K-1)j}^{j}(L) &= 2a	& \text{Rows } 2,7
	\end{align*}
Now, suppose $3a > \frac{T_1}{Q}$ and $Q>2$. (Both conditions occur for $Q$ large enough). Then, applying Remark~\ref{rem:tminusaplusa} regarding learning rates being larger when the arguments are farther away from one another (holding one fixed), the resulting rates with these candidates are dominated by (larger than) the rates between candidates $h$ and $i,j$, discussed next.

\parbold{Randomizing improves the minimum rate between candidates $h$ and $i,j$}
By Remark~\ref{rem:tminusaplusa}, 
\begin{align*}
r(T_1+a,T_2-a) >r(T_1,T_2)
\end{align*} 

Using $K$-Approval:
\begin{align*}
\text{Rate between $h,i$: }\,\,\,\,\,\,\,\,\,\,\ &r_{hi}(K) = r(T_1,T_2)\\
\text{Rate between $h,j$: }\,\,\,\,\,\,\,\,\,\,\ &r_{hj}(K) = r(T_1+a,T_2-a)\\
\text{Overall rate: }\,\,\,\,\,\,\,\,\,\,\ &r(K) = \min(r_{hi}(K),r_{hj}(K)) = r(T_1,T_2)
\end{align*}

Using $L$-Approval:
\begin{align*}
\text{Rate between $h,i$: }\,\,\,\,\,\,\,\,\,\,\ &r_{hi}(L) = r(T_1+a,T_2-a)\\
\text{Rate between $h,j$: }\,\,\,\,\,\,\,\,\,\,\ &r_{hj}(L) = r(T_1,T_2)\\
\text{Overall rate: }\,\,\,\,\,\,\,\,\,\,\ &r(K) = \min(r_{hi}(L),r_{hj}(L)) = r(T_1,T_2)
\end{align*}

Randomizing -- For any $0<p<1$, eliciting $K$-Approval with probability $p$, and $L$-Approval otherwise:
\begin{align*}
\text{Rate between $h,i$: }\,\,\,\,\,\,\,\,\,\,\ & r(T_1+(1-p)a,T_2-(1-p)a)\\
\text{Rate between $h,j$: }\,\,\,\,\,\,\,\,\,\,\  & r(T_1+pa,T_2-pa)\\
\text{Overall rate: }\,\,\,\,\,\,\,\,\,\,\ r(K) &= r(T_1 + \phi a,T_2 - \phi a) & \phi =\min(p,1-p)\\ %
&>r(T_1,T_2)  & \text{Remark }~\ref{rem:tminusaplusa}
\end{align*}
\parbold{$K'$-Approval ($K'\neq K, K' \neq L$) produces a worse rate than either $K$ or $L$ approval}

For any $K' < K-1$, $h$ is approved with probability $\epsilon_2 < \epsilon$, and $i,j$ are never approved. Then, the rate between $h$ and $i,j$ is $-\log(1 - \epsilon_2) \to 0$ as $\epsilon\to 0$. Identically, for $K'\geq L+2=K+3$, both $h$ and $i,j$ are approved except with some probability $\epsilon_2 < \epsilon$. 

For $K' = K-1$, $h$ is approved without $i$ with probability $a + \epsilon_2$ (for some $\epsilon_2 < \epsilon$), and $i$ is approved without $h$ with probability $a$. Then, the rate between $h$ and $i$ is $-\log(2\sqrt{(a + \epsilon_2)a} + 1 - 2a - \epsilon_2) \to 0$ as $\epsilon\to 0$.

For $K' = K+2 = L+1$, $h$ is approved without $i$ with probability $T_2 - a + \epsilon_2$, and $i$ is approved without $h$ with probability $0$. Then, the rate between them is $-\log(1 - T_2 + a - \epsilon_2)$. For $T_2$ small enough, this is a worse rate than using $K$ or $L$ approval. 

\parbold{The example described is asymptotically design-invariant}
From the above table, the probability that candidate $c\in\{1\dots M\}$ is in position $k$ or better, i.e., $\sigma(c)\leq k$ is:
\FloatBarrier
\begin{table}[h]
	\begin{tabular}{l|ccccccc}
		Candidate              & $k<K-1$ & $K-1$ & $K$ & $L=K+1$ & $K+2$ & $K+3$ & $M>k>K+3$ \\ \hline
	$w\in \{1 \dots K-2\}$		 &  $>0$       &   $>1-\epsilon$    & $1$    &  $1$       &   $1$    &   $1$    &    $1$     \\
	$ K-1$                 &   $>0$      &  $1-2a$     &  $1-2a$   &   $1-2a$      &   $1$    &   $1$    &    $1$     \\
	$K$                  &   $>0$      &  $>a$     &  $T_1 + a$   &     $T_1 + 3a$    &    $1$   &   $1$    &    $1$     \\\hline
	$L=K+1$                &    $0$     &  $a$     & $T_2$    &  $<T_2+\epsilon$      &   $<1$   &   $<1$    &    $<1$      \\
	$K+2$                  &    $0$     &   $0$    &   $T_2-a$  &    $<T_2+3a+\epsilon$     &   $<1$   &   $<1$    &    $<1$         \\
	$q\in\{K+3 \dots M\} $    &    $0$     &   $0$    &   $\frac{2a}{Q}$  &     $\frac{T_1 + 2T_2 -3a -\epsilon}{Q}$    &    $<1$   &   $<1$    &    $<1$    
\end{tabular}
\end{table}
\FloatBarrier
\parbold{Conditions on constants in problem}
For the above claims to hold, we set conditions on the constants in the problem. They are
\begin{align*}
0<\epsilon<a<\frac{T_2}{2}<T_2 &< T_1<T_1 + 2T_2 + a=1\\
1 - \epsilon &>2\sqrt{T_1T_2} + 1 - T_1 - T_2\\
\frac{T_1 + 2T_2 -3a -\epsilon}{Q} &< T_1 + 3a\\
3a &> \frac{T_1}{Q}\\
Q&>2\\
1 - T_2 + a - \epsilon &> 2\sqrt{T_1T_2} + 1 - T_1 - T_2
\end{align*}

This is a feasible set of constraints: $Q$ can be set large enough to meet conditions 3,4,5 for any fixed $T_1, a, T_2$ that meet condition 1. Condition 2 is weaker than the last condition. That leaves the last condition along with the first one. 
\begin{align*}
1 - T_2 + a - \epsilon &> 2\sqrt{T_1T_2} + 1 - T_1 - T_2\\
\iff a - \epsilon &> 2\sqrt{T_1T_2} - T_1 \\
\iff \frac{T_2}{2} - 2\epsilon &> 2\sqrt{T_1T_2} - T_1 & \text{set } \frac{T_2}{2}-\epsilon = a \\
\iff \frac{T_2}{2} + T_1 &> 2\epsilon + 2\sqrt{T_1T_2}
\end{align*}
which holds for $T_1$ large enough, and $T_2,\epsilon$ small enough. 

\end{proof}

%

%
%
%
%
%
%
%
%
%
%
%
%
%
%
%
%
%
%
%
%
%
%
%
%
%
%
%
%
%
%
%
%
%
%
%
%
%
%
%
%
%
%
%
%
%
%
%
%
%
%
%
%
%
%
%
%
%
%
%
%
%
%
%
%
%
%
%
%
%
%
%
%
%
%
%
%
%
%
%
%
%
%
%
%
%
%
%
%
%
%
%
%
%
%
%
%
%
%
%
%
%
%
%
%
%
%
%
%
%
%
%
%
%
%
%
%
%
%
%
%
%
%
%
%
%
%
%
%
%
%
%
%
%
%
%
%
%
%
%
%
%

%

%
%
%
%
%
%
%
%
%
%
%
%
%
%
%
%
%
%
%
%
%
%
%
%
%
%

\end{document}